\documentclass[journal]{IEEEtran}
\usepackage{amssymb}
\usepackage{amsthm}
\usepackage[reqno]{amsmath}
\usepackage{amsfonts}
\usepackage{bbm}
\usepackage{subfigure}
\usepackage{times}
\usepackage[final]{graphicx}
\usepackage{upgreek}
\usepackage{latexsym,amssymb,mathrsfs}
\usepackage{comment}
\usepackage{url}

\usepackage{algorithmicx}
\usepackage[ruled]{algorithm}
\usepackage{algpseudocode}
\usepackage{algpascal}
\usepackage{algc}
\usepackage{xcolor}
\usepackage{cite}
\newtheorem{thm}{Theorem}
\newtheorem{defi}{Definition}
\newtheorem{prop}{Proposition}
\newtheorem{lem}{Lemma}

\newtheorem{exmpl}{Example}

\begin{document}

\title{Context-aware Data Aggregation with Localized Information Privacy}

\author{\IEEEauthorblockN{Bo Jiang, \textit{Student Member, IEEE,}\quad
Ming Li, \textit{Senior Member, IEEE}\quad and Ravi Tandon, \textit{Senior Member, IEEE}
}

%\IEEEauthorblockA{Department of Electrical and Computer Engineering\\
%University of Arizona, Tucson, AZ, USA.\\}

%\IEEEauthorblockA{E-mail: \emph{\{bjiang, lim, tandonr\}}@email.arizona.edu}

\thanks{
This paper is an extension of our previous work in IEEE CNS2018\cite{Jian1805:Context}.

B. Jiang  M. Li and R. Tandon are with the Department
of Electrical and Computer Engineering, University of Arizona, Tucson,
AZ, 85742 USA e-mail: {\{bjiang, lim, tandonr\}}@email.arizona.edu}% <-this % stops a space
}

\markboth{Journal of \LaTeX\ Class Files,~Vol.~14, No.~8, August~2015}%
{Shell \MakeLowercase{\textit{et al.}}: Bare Demo of IEEEtran.cls for IEEE Journals}

\maketitle

\begin{abstract}
In this paper, localized information privacy (LIP) is proposed, as a new privacy definition, which allows statistical aggregation while protecting users' privacy without relying on a trusted third party. The notion of context-awareness is incorporated in LIP by the introduction of priors, which enables the design of privacy-preserving data aggregation with knowledge of priors. We show that LIP relaxes the Localized Differential Privacy (LDP) notion by explicitly modeling the adversary's knowledge. However, it is stricter than $2\epsilon$-LDP and $\epsilon$-mutual information privacy. The incorporation of local priors allows LIP to achieve higher utility compared to other approaches. For four different applications in privacy-preserving data aggregation, including survey, summation, weighted summation and histogram, we present an optimization framework, with the goal of minimizing the mean square error while satisfying the LIP privacy constraints. Utility-privacy tradeoffs are obtained under each model in closed-form, we then theoretically compare with the centralized information privacy and LDP. At last, we validate our analysis by simulations using both synthetic and real-world data. Results show that our LIP mechanism provides better utility-privacy tradeoffs than LDP and when the prior is not uniformly distributed, the advantage of LIP is even more significant. 
\end{abstract}


\IEEEpeerreviewmaketitle

\section{Introduction}

\IEEEPARstart{P}{rivacy} issues are crucial in this big data era, as users' data are collected both intentionally or unintentionally by a large number of private or public organizations. Most of the collected data are used for ensuring high quality of service, but may also put one's sensitive information at potential risk. For instance, when someone is rating a movie, his/her preferences may be leaked; when someone is searching for a parking spot nearby using a smartphone, his/her real location is uploaded and may be prone to leakage. To mitigate such privacy leakage, it is desirable to design privacy-preserving mechanisms that provide strong privacy guarantees without affecting data utility.

Traditional privacy notions such as  $k$-anonymity \cite{SS98} do not provide rigorous privacy guarantee and are prone to various attacks. On the other hand, Differential Privacy (DP) ~\cite{Dwork2008,Dwork20061} has become the $de facto$ standard for ensuring data privacy in the database community~\cite{Dwork2006}. The definition of DP assures each user's data has minimal influence on the output of certain types of queries on a database. In the classical DP setting, it is assumed that there is a trusted server which perturbs users' data while answering queries. However, more often then not, organization collecting users' data may not be trustworthy and the database storage system may not be secure \cite{Yahoo}. 

Recently, localized privacy protection mechanisms have gained attention as this setting allows local data aggregation while protecting each user's data without relying on a trusted third party \cite{aggregation}. In localized privacy-preserving data release, each user perturbs his or her data before uploading it; organizations that want to take advantage of users' data then aggregate with collected users' published results. { Earliest such mechanism is randomized response\cite{randomresponse}, which randomly perturbs each user's data. However, the original randomized response does not have formal privacy guarantees.
Later, Localized Differential Privacy (LDP) was proposed as a local variant of DP that quantifies the privacy leakage in the local setting\cite{Freudiger:2011:EPR:2186383.2186387}. Many schemes were proposed under the notion of LDP. For example, \cite{Extreme_ldp,ldp_lalitha,rr_ldp}, and Google's RAPPOR \cite{Rappor}.
LDP based data aggregation mechanisms have already been deployed in the real-world.} For example, in June 2016, Apple announced that it would deploy LDP-based mechanisms to collect user's typing data \cite{Apple}. However, Tang \textit{et al}. show that although each user's perturbation mechanism satisfies LDP, the privacy budget is too large ($\epsilon=43$)\footnote{The parameters, $\epsilon\ge{0}$, measures the privacy level. A smaller $\epsilon$ corresponds to a higher privacy level.} to provide any useful privacy protection. Wang \textit{et al}. provide a variety of LDP protocols for frequency estimation \cite{Tianhao} and compare their performance with Google's RAPPOR. However, for a given reasonable privacy budget, these protocols provide limited utility. Intuitively, compared with the centralized DP model, it is more challenging to achieve a good utility-privacy tradeoff  under the LDP model. The main reasons are two-fold: (1) LDP requires introducing noise at a significantly higher level than what is required in the centralized model. That is, a lower bound of noise magnitude of $\Omega_{\epsilon}(\sqrt{N})$ is required for LDP, where $N$ is the number of users. In contrast, only $O_{\epsilon}(1)$ is required for centralized DP \cite{Lowerbound}. 
(2)
LDP does not assume a neighborhood constraint on users' data as inputs, thus when the domain of data is very large, LDP leads to a significantly reduced utility \cite{Raef}.
 
In general, both localized and centralized DP provide strong context-free theoretical guarantees against worst-case adversaries \cite{context}. Context-free means that there is no knowledge of users' data (either instantaneous or statistical). On the other hand, context-aware privacy notions such as ones where statistical knowledge is available are favorable, as the utility can be increased by explicitly modeling the  adversary's knowledge. Information-theoretic privacy notions \cite{DBLP:journals/corr/WangYZ14b}\cite{ITP1}   that incorporate  statistical (prior) knowledge fall  into this category, which use mutual information (MI) to measure the information leaked about the original database in the released data \cite{ExtendingDP,7498650,MIP2}. Compared with context-free privacy notions, context-aware privacy notions, especially prior-aware notions achieve a better utility-privacy tradeoff \cite{context}.

We next discuss a simple illustrative example to motivate the need and advantages of context-aware privacy notions.

\begin{exmpl}
Consider taking a survey over $N=100$ individuals, where each person is independently asked whether he/she has been infected by some kind of disease. It is known based on clinical studies that this disease infects 1 out of 10 people on average. Each individual holds a local true answer $X_i$, where $i$ is the individual's index and $X_i=1$ if his/her answer is yes, $X_i=0$ if the answer is no. For privacy consideration, each individual perturbs his/her data by a randomized response mechanism (shown in Fig. \ref{fig:example_channel}) before publishing it. The goal is to estimate the aggregate $\sum_{i=1}^N X_i$ based on $Y_i, i=1,...,N$.

\begin{figure}[htp]
\centering
\includegraphics[width=3cm]{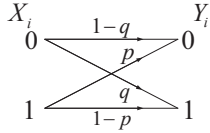}
\caption{Randomized response mechanism for each individual.}
\centering
\label{fig:example_channel}
\end{figure}
\end{exmpl}

 %\vspace{-10pt}
We first assume that the perturbation mechanism satisfies the context-free $\epsilon$-LDP for a given privacy budget $\epsilon$. By the definition of LDP, for each user: $\max\{\frac{p}{1-q},\frac{q}{1-p},\frac{1-q}{p},\frac{1-p}{q}\}\le{\epsilon}$. In \cite{Tianhao}, each user's input $X_i$ is treated as a fixed instance and a pair of valid solution for ($p, q$) is $p=q=\frac{1}{e^{\epsilon}+1}$. Using  the unbiased estimator adopted in \cite{Tianhao}, the  expected error of the aggregate is $\mathcal{E}_{LDP}=\frac{Np(1-p)}{(p-q)^2}=\frac{100e^{\epsilon}}{(e^{\epsilon}+1)^2}$.  The authors then derived an optimal solution for ($p, q$) by minimizing this error while subject to the privacy constraints. Their optimal values of $p$ and $q$ are different, where  $q^*=1-e^{\epsilon}/2$ and $p^*=0.5$, resulting in an expected estimation error of $\mathcal{E}_{Opt-LDP}=\frac{25}{(0.5-e^{\epsilon}/2)^2}$ in our example, which is  smaller than $\mathcal{E}_{LDP}$ (as shown in Fig. \ref{fig:example_variance}). This approach increases utility by  implicitly using prior knowledge, as it is based on the fact that the answers of a majority of users are zeros ($q^*$ is smaller than $p^*$). Unfortunately, the assumption that $X_i$, $i=1,2,...,N$ are instances rather than random variables prohibits introducing prior in a principled manner. In addition, the definition of LDP is independent of the priors, which is unable to adjust the perturbation parameters based on different priors. To explicitly introduce prior knowledge, a new privacy definition  is needed.
%Intuitively, this is because when $p$ and $q$ are close to 0.5, each individual's privacy is well protected, however the mechanism suffers decreased utility.  %We know that when $p$ and $q$ are close to 0, the mechanism guarantees little privacy protection while providing best utility; On the other hand, when $p$ and $q$ are close 0.5, each individual's privacy is well protected, however the mechanism suffers decreased utility. 

\setlength{\abovecaptionskip}{-0.05cm}
\begin{figure}[t]
\centering
\includegraphics[width=5cm]{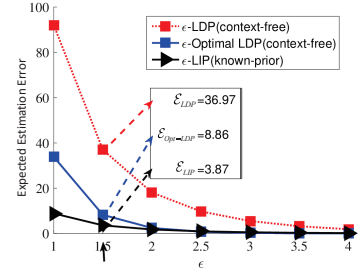}
\caption{Comparison among the expected estimation error of different approaches in Example 1:  considering prior in the perturbation mechanism significantly reduces the error.}
\centering
\label{fig:example_variance}
\end{figure}

 Assuming that each $X_i$ is a random variable with priors $P_1=Pr({X_i=1})$ and $P_0=Pr({X_i=0})$, we propose a context-aware LIP notion which imposes a bound on the ratio between the prior and posterior. In this example, we have: $-\epsilon\le{\{\frac{Pr(Y_i=0)}{1-q},\frac{Pr(Y_i=0)}{q},\frac{Pr(Y_i=1)}{1-p},\frac{Pr(Y_i=1)}{p}}\}\le{\epsilon}$.
 This notion guarantees that taking observations on the published data provides limited additional information of the real data. Restricted by this privacy notion, the optimal $p^*$ and $q^*$ which minimize the mean square error (MSE) can be derived as: $q^*=P_1/e^{\epsilon}$ and $p^*=P_0/e^{\epsilon}$. 
 Intuitively, it reduces MSE by carefully adjusting $p$ and $q$ to different priors. By this perturbation mechanism, the resulting error is $\mathcal{E}_{LIP}=\frac{10e^{\epsilon}-1}{0.64}$ which is significantly smaller than $\mathcal{E}_{LDP}$ (as shown in Fig. \ref{fig:example_variance}).

As we discussed above, introducing priors can provide higher utility. However, this comes at the overhead of estimating or learning prior knowledge. This can be obtained in two ways:
(1) Sometimes, each individual user's local prior is available (e.g., can be obtained by training based on historical published data)  \cite{Geo,7930028}. For instance, when Google wants to survey  multiple users' current locations to construct a traffic heat-map, it is  possible that it already possesses  past reported (unperturbed) locations of each user. 
(2)  However, local-prior is  relatively strong knowledge on users' data that may not always be attainable, and one may only be able to learn a global prior (assuming that users' data are identically distributed). For example, when taking a periodic survey, aggregated results in the recent past can be used as the global prior. Another example is when estimating the frequency of a rare disease,   one can leverage the results of past clinical research    to obtain a global prior \cite{bound}.

The  main contributions of this paper are three-fold:

(1)
We propose a new notion of Localized Information Privacy (LIP) for the local data release setting (without a trusted third party), which relaxes the notion of LDP by introducing priors to   increase data utility. We formally show that, $\epsilon$-LIP implies $\epsilon$-Mutual Information Privacy (MIP), and $\epsilon$-LIP is sandwiched between  $\epsilon$-LDP    and 2$\epsilon$-LDP.

(2) We apply the LIP notion to  privacy-preserving data aggregation. Focusing on four applications including: survey, (weighted) summation and histogram, we present a utility-privacy optimization framework with the goal of minimizing the mean squared error while  satisfying the LIP constraints. We deploy the randomized response type of perturbation mechanisms. We formulate  two perturbation and aggregation models, including: binary-input binary-output (BIBO) model, multiple-input multiple-output (MIMO) model, and derive the corresponding optimal perturbation parameters in closed form and discuss how each application can be realized. We show that one set of the optimal solution can minimize the MSE and the Mean absolute error (MAE) between the raw data and the perturbed data simultaneously. We theoretically demonstrate the advantages of the proposed mechanism and compare with the LDP mechanisms. For comparison, we also study a centralized version of data  aggregation under information privacy (CIP), and derive its optimal utility-privacy tradeoff. 

(3) We validate our analysis using simulations on  both synthetic and real-world datasets (i.e., Karosak, a website-click stream data set, and Gowalla, a location aggregation data set). Both theoretical and simulation results  show that optimal perturbation mechanisms under  $\epsilon$-LIP always achieve  a better utility-privacy tradeoff  than those under $\epsilon$-LDP when $\epsilon>0$, especially when the prior is not uniformly distributed. In the MIMO case, We show that when the domain size increases, the advantage of the LIP is enhanced. We show that the advantages of the context-aware LIP are two-fold: on one hand, users are setting the perturbation parameters according to the prior knowledge; on the other hand, the adversary is allowed to build prior-related estimators that can minimize the expected errors. When compared with CIP,  it always leads to a better utility-privacy tradeoff than LIP,  as the server possesses global information. However, the utility gap is not large. 

The remainder of the paper is organized as follows. In Section~\ref{sec:privacy_model}, we describe the proposed LIP notion and its relationship with other existing privacy notions. In Section~\ref{sec:model}, we introduce the system model and problem formulation. In Section~\ref{tradeoff}, we derive the utility-privacy tradeoff under the several applications. In Section~\ref{sec:sim}, we present the simulation results and   compare utility-privacy tradeoffs among different models. In Section~\ref{sec:con}, we offer concluding remarks and discuss  future directions.

\section{Privacy Definitions}\label{sec:privacy_model}
In local privacy-preserving data release, each user  uploads its perturbed data directly to an untrusted aggregator. In this section,  we first recap two existing privacy notions in localized settings, and then present our new LIP definition.

The traditional LDP definition guarantees that each user's perturbed data has a similar probability to result in the same output for any two inputs from the data domain $\mathbb{D}$:

\begin{defi}($\epsilon$-Localized Differential Privacy (LDP))\cite{NIPS2014_5392}\label{def:LDP}
A mechanism $\mathcal{M}$ which takes input $X$ and outputs $Y$ satisfies  $\epsilon$-LDP for some $\epsilon\in{\mathbb{R}^+}$, if $\forall{x, x'\in{\mathbb{D}}}$ and $\forall{y\in{Range(\mathcal{M})}}$: 
\begin{equation}
    \frac{Pr(\mathcal{M}(x)=y)}{Pr(\mathcal{M}(x')=y)}\le{e^{\epsilon}},
\end{equation}
\end{defi}

LDP provides strong context-free privacy guarantee against worst-case   adversaries.  However, there are many scenarios where some context of $X$ is available (e.g., prior distribution). In such situations,   introducing  context   provides  relaxed privacy guarantees. 
One such definition is mutual information privacy, which uses the mutual information between $Y$, $X$ to measure the average information leakage of $X$ contained in $Y$: 

\begin{defi}($\epsilon$-Mutual Information Privacy (MIP))\cite{7498650}
A mechanism $\mathcal{M}$ which takes input $X$ and outputs $Y$, satisfies $\epsilon$-MIP for some $\epsilon\in{\mathbb{R}^+}$, if the mutual information between $X$ and $Y$ satisfies $I(X;Y)\le{\epsilon}$, where $I(X;Y)$ is: 
\begin{equation}\label{mutualinfo}
    \sum_{x,y\in{\mathbb{D}}}Pr(X=x, Y=y)\log\frac{Pr(X=x,Y=y)}{Pr(X=x)Pr(Y=y)}.
\end{equation}
\end{defi}

Originally, MIP was proposed under the centralized setting where $X$ is the database or individual items and $Y$ is a query output. Here we can adapt it to the local setting, where $X$ and $Y$ are each individual user's input and output.

Although MIP is  context-aware, it is a relative weak privacy notion since it only bounds   the average information leakage. There may exist some $(x,y)$ pair that makes the ratio between the joint and product of marginal distributions very large (while the  joint probability is very small). 

In order to  limit the information leakage of every pair of realizations of $X$ and $Y$, we consider a bound on the ratio between the prior $Pr(X)$ and posterior $Pr(X|Y)$, which leads to our proposed localized information privacy notion:

\begin{defi}($\epsilon$-Localized Information Privacy (LIP))\label{def:LIP}
A mechanism $\mathcal{M}$ which takes input $X$ and output $Y$ satisfies $\epsilon$-LIP for some $\epsilon\in{\mathbb{R}^+}$, if $\forall{x,y\in{\mathbb{D}}}$:
\begin{equation}\label{eq1}
    e^{-\epsilon}\le{\frac{Pr(X=x)}{Pr(X=x|Y=y)}}\le{e^{\epsilon}}.
\end{equation}
\end{defi}

 Intuitively,   LIP guarantees that having the knowledge of users' priors, the adversary can't infer too much additional information about each input $x$ by observing each output $y$. Note that,  when $\epsilon$ is small, this ratio   is bounded close to 1. Definition (\ref{def:LIP}) can be viewed as the localized version of information privacy, which focused on a centralized setting \cite{Centrlized_IP}, and the main  differences are in the  definitions of input and output. Again, here $X$ and $Y$ stand for each user's input and output variables, respectively. 

In the following, We   show that LIP is a stronger privacy notion than MIP, since the latter only provides an average privacy guarantee while LIP bounds the leakage on every pair of realizations of  $X$ and $Y$.
\begin{prop}
If a mechanism $\mathcal{M}$ satisfies $\epsilon$-LIP, it also satisfies $\epsilon$-MIP.
\end{prop}
\begin{proof}
Assume that $\mathcal{M}$ satisfies $\epsilon$-LIP. By Bayes rules, we have that, $\forall x, y\in{\mathbb{D}}$:
\begin{equation}\label{eq3}
\setlength{\abovedisplayskip}{3pt}
\setlength{\belowdisplayskip}{3pt}
    e^{-\epsilon}\le{\frac{Pr(X=x,Y=y)}{Pr(X=x)Pr(Y=y)}}\le{e^{\epsilon}}.
\end{equation}

Substituting \eqref{eq3} into \eqref{mutualinfo}, we get:
\begin{equation*}
\begin{aligned}
\setlength{\abovedisplayskip}{3pt}
\setlength{\belowdisplayskip}{3pt}
&I(X,Y)\le{\epsilon\sum_{x,y\in{\mathbb{D}}}Pr(x,y)}=\epsilon,
\end{aligned}
\end{equation*}
where $\sum_{x,y\in{\mathbb{D}}}Pr(x,y)=1$.\end{proof}

Furthermore, the following theorem shows the relationship between   LIP and LDP:
\begin{thm}
If a mechanism $\mathcal{M}$ satisfies $\epsilon$-LIP, then it also satisfies $2\epsilon$-LDP;
if a mechanism $\mathcal{M}$ satisfies $\epsilon$-LDP, then it also satisfies $\epsilon$-LIP.
\end{thm}
\begin{proof}
To prove the first part, consider a mechanism $\mathcal{M}$  that takes any two inputs $X=x$, $X=x'$ and outputs the same $Y=y$.

When $\mathcal{M}$ satisfies $\epsilon$-LIP, using the Bayes rule, Definition (\ref{def:LIP}) is equivalent to:
 \begin{equation}\label{eq2}
 \setlength{\abovedisplayskip}{3pt}
\setlength{\belowdisplayskip}{3pt}
    e^{-\epsilon}\le{\frac{Pr(Y=y)}{Pr(Y=y|X=x)}}\le{e^{\epsilon}}. 
\end{equation}

Since the above also holds for  $X=x'$, we have:
 \begin{equation}
 \setlength{\abovedisplayskip}{3pt}
\setlength{\belowdisplayskip}{3pt}
    e^{-\epsilon}\le{\frac{Pr(Y=y)}{Pr(Y=y|X=x')}}\le{e^{\epsilon}}.
\end{equation}
Inequality (\ref{eq2}) is equivalent to:
 \begin{equation*}
 \setlength{\abovedisplayskip}{3pt}
\setlength{\belowdisplayskip}{3pt}
    e^{-\epsilon}\le{\frac{Pr(Y=y|X=x)}{Pr(Y=y)}}\le{e^{\epsilon}}.
\end{equation*}
Since both of the metrics in above inequalities are positive, by multiplying these two inequalities, we   get:
\begin{equation*}
\setlength{\abovedisplayskip}{3pt}
\setlength{\belowdisplayskip}{3pt}
   e^{-2\epsilon}\le\frac{Pr(Y=y|X=x)}{Pr(Y=y|X=x')}\le{e^{2\epsilon}}.
\end{equation*}
Since we can switch $x$ and $x'$, it is equivalent to the definition of $2\epsilon$-LDP.

To prove the second part, if $\mathcal{M}$ satisfies $\epsilon$-LDP, we have:
\begin{equation*}
\setlength{\abovedisplayskip}{3pt}
\setlength{\belowdisplayskip}{3pt}
   Pr(Y=y|X=x')\le{e^{\epsilon}}Pr(Y=y|X=x).
\end{equation*}
On the other hand:
\begin{equation*}
\begin{aligned}
Pr(Y=y)=&\sum_{x'\in{\mathbb{D}}}Pr(Y=y|X=x')Pr(X=x')\\
\le&e^{\epsilon}Pr(Y=y|X=x)\sum_{x'\in{\mathbb{D}}}Pr(X=x')\\
=&e^{\epsilon}Pr(Y=y|X=x),
\end{aligned}
\end{equation*}
by switching inputs, we can also get:
\begin{equation*}
\setlength{\abovedisplayskip}{3pt}
\setlength{\belowdisplayskip}{3pt}
   Pr(Y=y)\ge{e^{-\epsilon}}Pr(Y=y|X=x),
\end{equation*}
which means that $\mathcal{M}$ also satisfies $\epsilon$-LIP.
\end{proof}

\begin{table}[t]
\caption{List of symbols}
\footnotesize
\centering
\begin{tabular}{c{l} c}
\hline
$\mathbb{D}$ & The universe of input values\\
$X$ & Input random variable\\
$P$ & User's prior distribution\\
$Y$ & Output random variable\\
$\bar{X}$ &  Set of input data\\
$\bar{Y}$ & Set of output data\\
$i$ & Index User ($i\in{\{1,2,3,...,N\}}$)\\
$N$ & Total number of users\\
$\mathcal{M}$ & Privacy preserving mechanism\\
$\mathbf{q}$ & Set of perturbation parameters\\
$f(\cdot)$ & Aggregation function\\
$\hat{X}$ & Aggregated data of user\\
$\hat{\mathbf{S}}$ & Minimized mean square error estimation\\ 
$\epsilon$ & Privacy budget\\
$U$ & Utility measurement\\
$\mathcal{E}$ & Mean square error function\\
$\mathcal{T}$ & Feasible region of $\mathbf{q}$ \\
\hline
\end{tabular}
\end{table}

Thus, $\epsilon$-LIP is a more relaxed privacy notion than $\epsilon$-LDP. However, it is stronger than $2\epsilon$-LDP. Intuitively, LIP relaxes LDP because LDP results in the same output $y$ for every input $x$, no matter what his/her prior is. On the other hand, for inputs with different priors, LIP perturbs differently. For example, when a user with $Pr(X=1)=0.99$, if he holds $X=1$, which means his real data is consistent with the prior knowledge. As LIP bounds on the ratio between prior and posterior, it has a large probability to output $1$ to make the posterior probability similar to its prior; if he holds $X=0$, which has a small prior to happen, he also has large probability to output 1 to make the posterior probability small.

\section{Models  and Problem Formulation}\label{sec:model}
\subsection{System and Threat Models}

\begin{figure}[t]
\centering
\includegraphics[width=8cm]{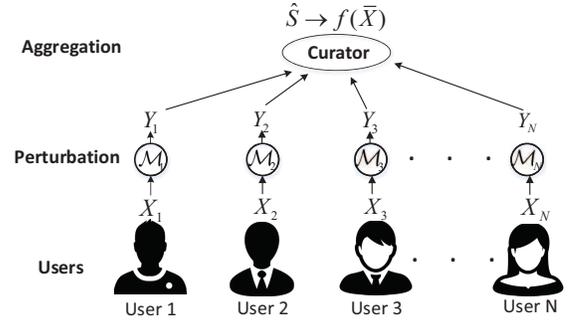}
\caption{System Model of Privacy-Preserving Data Aggregation.}
\vspace{-10pt}
\centering
\label{fig:Model}
\end{figure}

Consider a data aggregation system with  $N$ users and a data curator. Each user $i$ locally generates  private data which is denoted as random variable $X_i$,  taking value $x_i$ from the domain $\mathbb{D}=\{a_1,a_2...a_d\}$ with probability $P^i_k=Pr(X_i=k)$. We assume that $X_i$s  are independent from each other. Before publishing his/her data to the curator, each user locally perturbs it by a privacy-preserving mechanism $\mathcal{M}_i$. The output is denoted as  $Y_i$  which takes value $y_i$ from $\mathbb{D}$. The mechanism $\mathcal{M}_i$   maps each possible input to each possible output with certain probability, the set of them are called perturbation parameters (denoted as $\mathbf{q}^i$). After receiving each user's perturbed data, the curator computes a statistical function on those data (for example, estimating the frequency of certain input which will be useful for data mining).  The  system model is depicted in Fig. \ref{fig:Model}.

%It is assumed that each user in the aggregation shares a same privacy budget ($\epsilon$). 

The curator is considered as untrusted due to both internal  and external threats. On the one hand, users' private data is profitable  and companies can be interested in user tracking or selling their data. On the other hand, data breaches may happen from time to time due to hacking activities.   Denote the true aggregated result by $f(\bar{X})$, where $\bar{X}=\{X_1,X_2,...,X_N\}$. The curator (adversary) observes all the users' perturbed outputs $\bar{Y}=\{Y_1,Y_2,...,Y_N\}$ and tries to obtain an estimate of $\mathbf{S}=f(\bar{X})$. Furthermore, we assume that the adversary  possesses knowledge of  prior distributions of   users' inputs, and the algorithms/perturbation mechanisms that users adopt to publish their data. The  curator aims at performing accurate estimations using all the information above, but is also interested in inferring each user's real input   $X_i$.

%Using the MMSE estimator $\hat{\mathbf{S}}=E[\mathbf{S}|\bar{Y}]$. 

For different applications of data aggregation, the definition of $f(\cdot)$s varies. In this paper, four most common applications are considered: 
\begin{itemize}
\item Survey: the curator is interested in estimating how many people in the survey  possess a private value $v$. Each user's data can be mapped into a binary bit as: $f_i(X_i)=\mathbbm{1}_{\{v={X_i}\}}$, where $\mathbbm{1}_{\{a=b\}}$ is an indicator function, which is 1 if $a=b$; 0 if $a\neq{b}$. Then the curator just sum up all the collected data to get the answer. Thus $\mathbf{S}=f(\bar{X})=\sum^N_{i=1}\mathbbm{1}_{\{v={X_i}\}}$.
%For example, the census bureau   wants to know how many people believe in a particular religion.  Different users may have different answers. According to the particular religion that the curator is interested in, each user then maps his/her true data as 1 if he/she believes in that religion, otherwise it is 0\cite{Tianhao}. 
\item Summation: Summation results are usually used to measure an average property of the surveyed individuals. For example, the census bureau wants to survey the average income over a group of people, mathematically, $\mathbf{S}=f(\bar{X})=\frac{1}{N}\sum^N_{i=1}{X_i}$. 
\item Weighted summation (linear combination): It's straightforward to extend the direct summation to a weighted summation (plus an offset) for more general applications. For example, assume that the curator is interested in some particular users more than others, such as employer v.s. employees; adults v.s. children; the professionals v.s. amateurs. Thus, these important users' data are assigned a larger coefficient than others. The offset can be used as a correction to the raw data. For example, avoiding summation blowing up or controlling data in a particular range. Thus $\mathbf{S}=f(\bar{X})=\sum^N_{i=1}{(a_iX_i+b_i)}$.  
\item  Histogram: In this application, the curator is interested in estimating how many people possess each of the data category in $\mathbbm{D}$, or classify people according to their data value. For example, the curator wants to statistically estimate a location frequency matrix with each entry standing for how many people are within a particular location. For this application, $\mathbf{S}$ is a set of "categorized" data: $\{S_{1},S_{2},...,S_{d}\}$, such that, $\forall{{k}\in\mathbb{D}}$, $S_{k}=\sum^N_{i=1}\mathbbm{1}_{\{X_i=k\}}$.
\end{itemize}
%Note that,  the curator may have incomplete/less accurate knowledge about each user's local priors. Intuitively, it will reduce the estimation accuracy, however it also mitigates privacy leakage. The utility-privacy tradeoff under such scenarios remain as our  future work.

\subsection{Privacy and Utility Definitions}

%\textbf{Definition of privacy}:
The \textit{privacy}   of each user's input satisfies LIP and is parameterized by the privacy budget ($\epsilon$) in Definition (\ref{def:LIP}). The smaller $\epsilon$ is, the higher privacy level the mechanism satisfies. Note that, for simplicity, we consider $\epsilon$ to be the same for all the users; however it is straightforward to extend our model and results to different $\epsilon$ for different users. Under LIP, the privacy constraints can be formulated as: $\forall i\in{\{1,2,...,N\}}$ and $\forall x_i,y_i\in\mathbb{D}$, there is 

\begin{small}
\begin{equation}\label{eqlip2}
\setlength{\abovedisplayskip}{3pt}
\setlength{\belowdisplayskip}{3pt}
e^{-\epsilon}\le\frac{Pr(Y_i=y_i|X_i=x_i)}{Pr(Y_i=y_i)}\le{e^{\epsilon}}.
\end{equation}
\end{small}

Note that, $\mathbf{q}^i\triangleq\{Pr(Y_i=y_i|X_i=x_i), \forall x_i, y_i\in\mathbb{D}\}$. When $\epsilon$ is given, the set of inequalities in Eq. (\ref{eqlip2}) constrains $\mathbf{q}^i$ to be within  a feasible region $\mathcal{T}_i$,  $\forall{i\in{1,2,..N}}$.

 The definition of \textit{utility} depends on the application scenario. For example, in statistical aggregation, estimation accuracy is often measured by absolute error or mean square error \cite{MMSE}\cite{heavyhitter}; in location tracking, it is typically measured by Euclidean distance \cite{Geo}; in privacy-preserving data publishing, distortion is usually used to measure the utility \cite{7498650}.

Focus on the four applications discussed above, we define utility as the inverse of the Mean Square Error (MSE): 
$U(\mathbf{S},\hat{\mathbf{S}})=-\mathcal{E}(\mathbf{S},\hat{\mathbf{S}})$, where $\mathcal{E}(\mathbf{S},\hat{\mathbf{S}})=E[(\mathbf{S}-\hat{\mathbf{S}})^2]$, $\hat{\mathbf{S}}$ is the estimated $\mathbf{S}$. Notice that, given $P^i$, $\mathcal{E}(\mathbf{S},\hat{\mathbf{S}})$ depends only on each user's perturbation parameters: $\mathbf{q}^1,..., \mathbf{q}^N$, as any estimator $\hat{\mathbf{S}}$ will depend on the output $\bar{Y}$ whose distribution is a function of $\mathbf{q}^i$. Thus, maximizing  the utility  is equivalent to find the optimal parameters to minimize the MSE.

\subsection{Problem Formulation}\label{sec:problem}

  In general, there is a tradeoff between utility and privacy. We can   formulate  the following   optimization problem  to find the optimal perturbation mechanism that yields the optimal tradeoff:  
\begin{equation}\label{utility}
\begin{aligned}
&\quad\quad\quad\min \mathcal{E}(\mathbf{q}^1,..., \mathbf{q}^N),\\
 &s. t. ~~~\mathbf{q}^i\in{\mathcal{T}_i},~~ \forall{i\in{1,2,..N}}. 
 \end{aligned}
\end{equation}

From \cite{papoulis2002probability}, it is well known that the optimal estimator that results in the minimized mean square error (MMSE) is $\hat{\mathbf{S}}=g(\bar{Y})=E[\mathbf{S}|\bar{Y}]$. Since $E[E[\mathbf{S}|\bar{Y}]]=E[\mathbf{S}]$, $\hat{\mathbf{S}}$ is an unbiased estimator. Therefore, we use the MMSE estimator in Eq. (\ref{utility}).

\textbf{Problem Decomposition:} Next, we show how the problem defined in Eq. (\ref{utility}) can be decomposed into the local cases. 

Since  we assume that each user's input is independent from each other, all the $f(\cdot)$ functions above can be decomposed into local functions of each $X_i$: local functions in application 1 and application 4 are indicator functions (or vector);  local functions in application 3 is a linear function. Without loss of generality, we denote the local function for user $i$ as $f_i$.

In the local setting, users independently perturb their data, thus each of them results in a MSE in aggregation, which is denoted by $\mathcal{E}_i=E[(f_i(X_i)-E[f_i(X_i)|Y_i])^2]$ (for the fourth application, denote $\mathcal{E}^k_i=E[(f^k_i(X_i)-E[f^k_i(X_i)|Y_i])^2]$) as the MSE of user $i$ when aggregating the $k$-th data, and the overall utility defined in Eq. \eqref{utility} satisfies decomposition proposition:

\begin{figure*}[htp]
\centering 
\subfigure[The binary model: different local priors and asymmetric perturbation] 
{ \includegraphics[width=5cm]{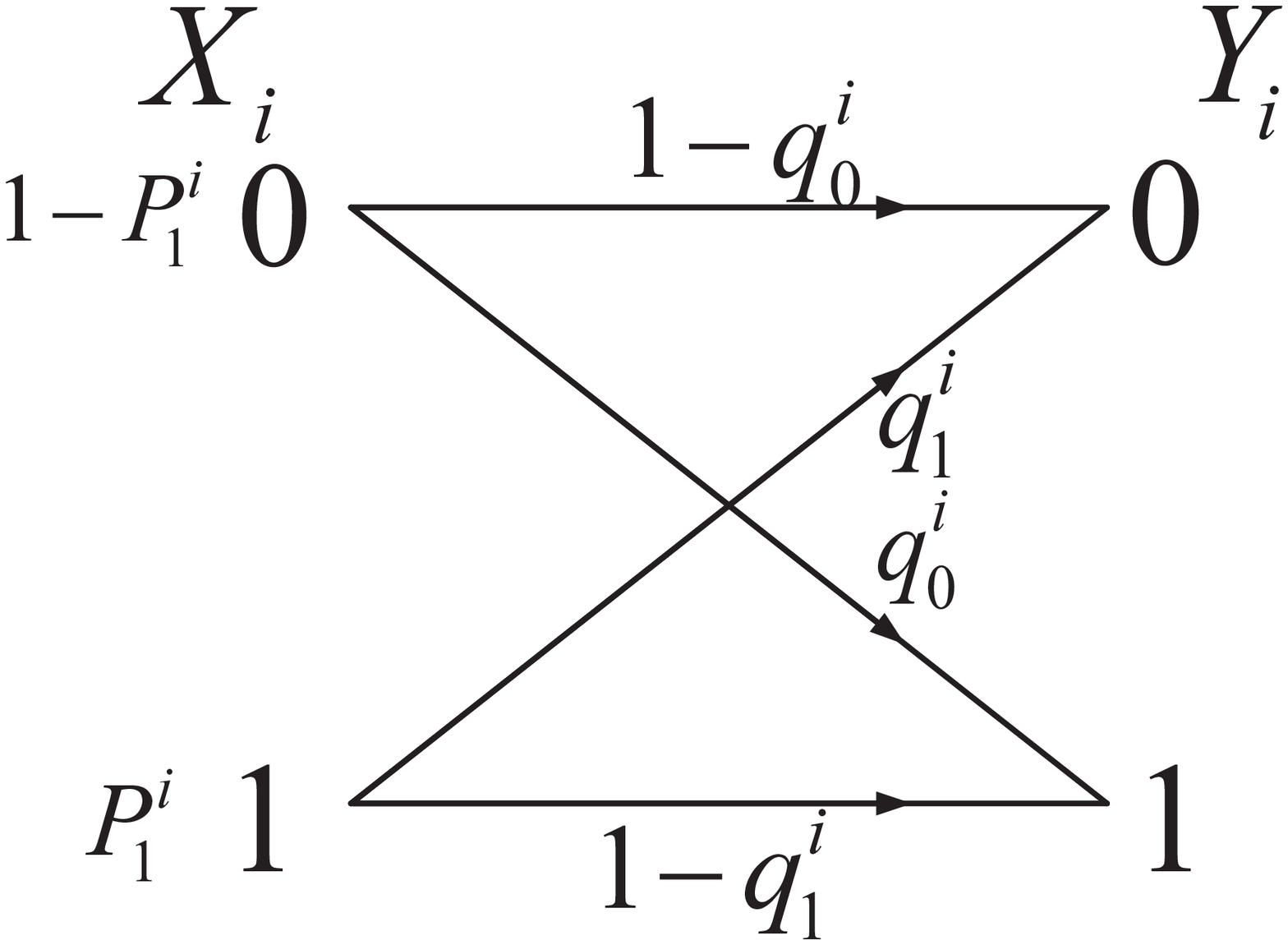} 
\label{fig:General_Model} } 
\subfigure[Model with multiple input and multiple output] 
{ \includegraphics[width=5cm]{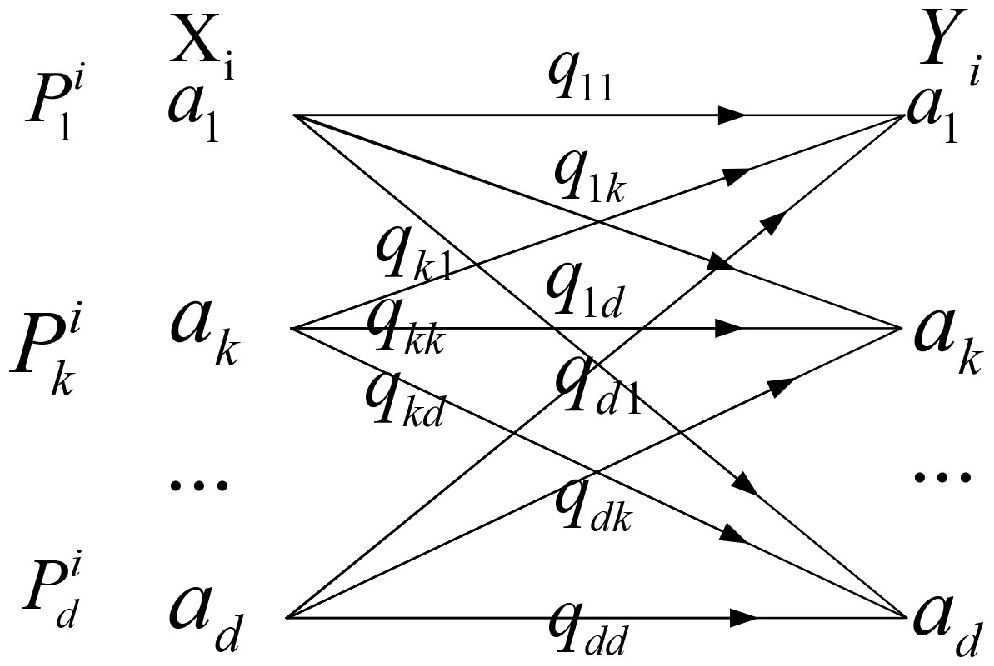} 
\label{fig:MIMO_Model} } 
\subfigure[Centralized binary information privacy model (trusted server based)] 
{ \includegraphics[width=5cm]{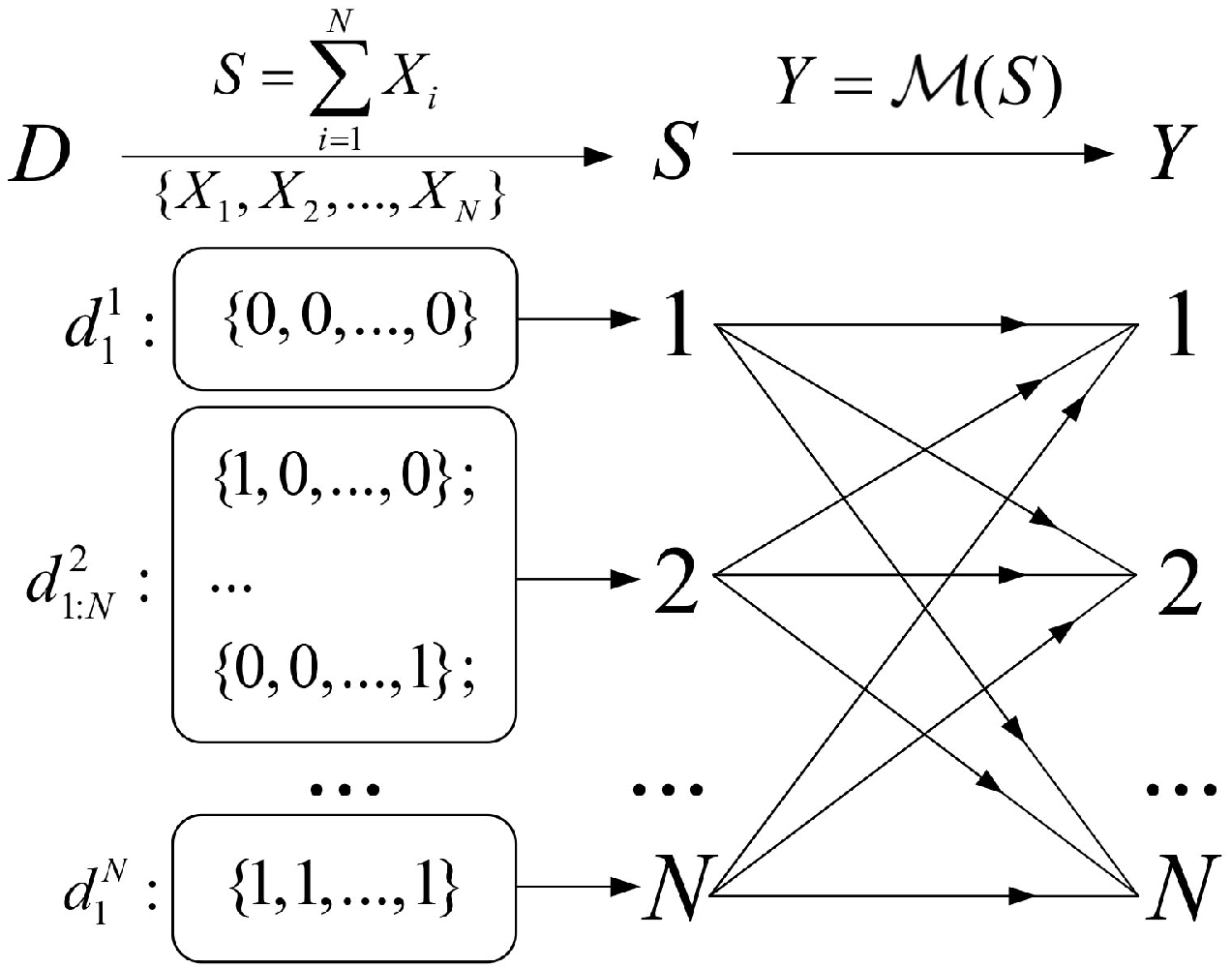} 
\label{fig:Central} } 
\caption{Different models for the perturbation mechanism considered in this paper ((a) and (b) are for the $i$-th user).} 
\label{model_compare} 
\end{figure*}

\begin{prop}\label{thm:decompose}
The global optimization problem defined in Eq. \eqref{utility} can be decomposed into $N$ local optimization problems, under independent user inputs.
\begin{equation}
\setlength{\abovedisplayskip}{3pt}
\setlength{\belowdisplayskip}{3pt}
\min_{(\mathbf{q}^i)\in{\mathcal{T}_i}}\mathcal{E}(\mathbf{q}^1,..., \mathbf{q}^N)=\sum_{i=1}^N\min_{(\mathbf{q}^i)\in{\mathcal{T}_i}}\mathcal{E}_i(\mathbf{q}^i).
\end{equation}
\end{prop}
\begin{proof}

Observe that, the $f(\cdot)$s in the four  basic applications above can be expressed as a summation over all the $f_i(X_i)$s, as the semantic of aggregation implies a summation operation. thus the summation based MMSE estimator $\hat{\mathbf{S}}$ can be expressed as:
\begin{small}
\begin{equation}\label{estimator}
\setlength{\abovedisplayskip}{3pt}
\setlength{\belowdisplayskip}{3pt}
\begin{aligned}
&E[\mathbf{S}|\bar{Y}]=E[f(\bar{X})|\bar{Y}]
=E[\{f(X_1,X_2,...,X_N)\}|\bar{Y}]\\
\overset{(a)}{=}&E[f_1(X_1)|\bar{Y}]+E[f_2(X_2)|\bar{Y}],...,+E[f_N(X_N)|\bar{Y}]\}\\
\overset{(b)}{=}&\sum^{N}_{i=1}\{E[f_i(X_i)|Y_i]\}.
\end{aligned}
\end{equation}
\end{small}

where (a) in Eq. \eqref{estimator} is due to the independence of $X_i$s, and (b) is because $X_i$ is only correlated with $Y_i$ in the output sequence.  Thus $\mathcal{E}({\mathbf{S},\hat{\mathbf{S}}})$  can be derived as:
\begin{equation}\label{error_1}
\mathcal{E}({\mathbf{S},\hat{\mathbf{S}}})=E[(\sum^{N}_{i=1}\{f_i(X_i)-E[f_i(X_i)|Y_i]\})^2].
\end{equation}

Note that, for the application of histogram, the error forms a error vector that $(S_k,\hat{S_k})_{k=1}^d$. By the definition of second order norm. the mean square error of this case is: 
\begin{equation}\label{error_11}
\mathcal{E}(S_k,\hat{S_k})_{k=1}^d=\sum^d_{k=1}E[(\sum^{N}_{i=1}\{f^k_i(X_i)-E[f^k_i(X_i)|Y_i]\})^2],
\end{equation}
where $f^k_i(X_i)=\mathbbm{1}_{\{X_i=k\}}$.

For the first three applications, we next show that the total MSE can be decompose into the summation of local MSEs. (the proof for histogram is shown in Appendix.D.
\begin{small}

\begin{equation*}
\setlength{\abovedisplayskip}{3pt}
\setlength{\belowdisplayskip}{3pt}
\begin{aligned}
&\mathcal{E}({\mathbf{S},\hat{\mathbf{S}}})=E[(\sum^{N}_{i=1}\{f^k_i(X_i)-E[f^k_i(X_i)|Y_i]\})^2]\\
&=(\sum^{N}_{i=1}E\{f^k_i(X_i)-E[f^k_i(X_i)|Y_i]\})^2\\
&-2\sum^{N}_{j=1,l\neq{j}}E\{(f^k_j(X_i)-E[f^k_j(X_i)|Y_j])(f^k_l(X_i)-E[f^k_l(X_i)|Y_l])\}.\\
\end{aligned}
\end{equation*}
\end{small}
The cross terms equal to 0 because $\forall{j,l}\in\{1,N\}$ and $j\neq{l}$:
\begin{small}

\begin{equation}\label{eq:crossterm}
\begin{aligned}
&E\{(f^k_j(X_i)-E[f^k_j(X_i)|Y_j])(f^k_l(X_i)-E[f^k_l(X_i)|Y_l])\}]\\
=&E[(f^k_j(X_i)-E[f^k_j(X_i)|Y_j])]E[(f^k_l(X_i)-E[f^k_l(X_i)|Y_l])]\\
=&[E(f^k_j(X_i))-E\{E[f^k_j(X_i)|Y_j]\}][E(f^k_l(X_i))-E\{E[f^k_l(X_i)|Y_l]\}].
\end{aligned}
\end{equation}

\end{small}
In \eqref{eq:crossterm}, $E(f^k_j(X_i))-E\{E[f^k_j(X_i)|Y_j]\}$ and $E(f^k_l(X_i))-E\{E[f^k_l(X_i)|Y_l]\}$ are 0, thus, $\mathcal{E}({\mathbf{S},\hat{\mathbf{S}}})=\sum^N_{i=1}\mathcal{E}^k_i(\mathbf{q}^i)$ 
%and 
%$\mathcal{E}(S_k,\hat{S_k})^{d}_{k=1}=\sum^{d}_{k=1}\sum^N_{i=1}\mathcal{E}^k_i(\mathbf{q}^i)$.

We next show that the overall optimal solutions (perturbation parameters) satisfy each local privacy constraints: 

Assume that for each user, the minimized $\mathcal{E}_i(\mathbf{q}^i)=e_i$ is achieved at $\mathbf{q}^{i*}\in{\mathcal{T}_i}$, then $\mathcal{E}(\mathbf{q}^{1*},...,\mathbf{q}^{N*})=\sum_{i=1}^Ne_i$. 

If for some $user_k$ who takes parameters $\mathbf{q}^{k}\in{\mathcal{T}_k}$, by assumption, we know that $\mathcal{E}_k(\mathbf{q}^k)\ge{e_k}$. Thus
\begin{equation*}
\setlength{\abovedisplayskip}{3pt}
\setlength{\belowdisplayskip}{3pt}
\sum_{i=1}^k\mathcal{E}_i(\mathbf{q}^{i*})+\mathcal{E}_k(\mathbf{q}^k)+\sum_{i=k+1}^N\mathcal{E}_i(\mathbf{q}^{i*})\ge{\sum_{i=1}^Ne_i}.
\end{equation*}
That means the minimal value of $\mathcal{E}(\mathbf{q}^1,...,\mathbf{q}^N)$, where $\mathbf{q}^i\in{\mathcal{T}_i}$, 
$\forall{i\in{[1,N]}}$ can be achieved if for each user, $\mathbf{q}^i=\mathbf{q}^{i*}$.

%Note that, although $\mathcal{E}_i$s that are inside a category are independent with each other, the errors of different categories (such as $\mathcal{E}(S_j,\hat{S_j})$ and $\mathcal{E}(S_l,\hat{S_l})$, where $l\neq{j}\in{\mathbb{D}}$) are dependent, thus whether the decomposition proposition works for the statistical aggregation is not a direct result, which can be proved using the optimal perturbation parameters. We discuss this case into detail in section \ref{model_applications}
\end{proof}
%\begin{proof}
%The MSE defined in Eq. \eqref{error_1} can be expressed as:
%\begin{equation}\label{eq_error}
%\setlength{\abovedisplayskip}{3pt}
%\setlength{\belowdisplayskip}{3pt}
%\begin{aligned}
%&\sum^N_{i=1}\mathcal{E}_i(\mathbf{q}^i)
%+2E[\sum^N_{k\neq{j}}(X_k-E[X_k|Y_k])(X_j-E[X_j|Y_j])].\\
%\end{aligned}
%\end{equation}
%As users are assumed to be independent, the expected cross terms in Eq. \eqref{eq_error} are 0s. Thus $\mathcal{E}(\mathbf{q}^1,..., \mathbf{q}^N)=\sum^N_{i=1}\mathcal{E}_i(\mathbf{q}^i)$.

%Assume that for each user, the minimized local MSE $\mathcal{E}_i(\mathbf{q}^i)=e_i$ can be achieved at $\mathbf{q}^{i*}$, where $\mathbf{q}^{i*}\in{\mathcal{T}_i}$, then $\mathcal{E}(\mathbf{q}^{1*},...,\mathbf{q}^{N*})=\sum_{i=1}^Ne_i$. 

%If for some $user_k$ who takes parameters $\mathbf{q}^{k}\in{\mathcal{T}_k}$, by assumption, we know that $\mathcal{E}_k(\mathbf{q}^k)\ge{e_k}$. Thus
%$$\sum_{i=1}^k\mathcal{E}_i(\mathbf{q}^{i*})+\mathcal{E}_k(\mathbf{q}^k)+\sum_{i=k+1}^N\mathcal{E}_i(\mathbf{q}^{i*})\ge{\sum_{i=1}^Ne_i}.$$
%That means the minimal value of $\mathcal{E}(\mathbf{q}^1,...,\mathbf{q}^N)$, where $\mathbf{q}^i\in{\mathcal{T}_i}$, 
%$\forall{i\in{[1,N]}}$ can be achieved if for each user, $\mathbf{q}^i=\mathbf{q}^{i*}$.
%\end{proof}
By proposition \ref{thm:decompose}, when the perturbation parameters of each user are optimal, the overall MSE of the mechanism achieves its minimum. In addition, each user can perform its local optimization independent from each other, which well suits the local setting.

Notice that the resulted MSE by a  MMSE estimator for each user is $\mathcal{E}_i(\mathbf{q}^i)=E[Var(f_i(X_i)|Y_i)]$.
By the law of total variance:
 \begin{equation}\label{largerlaw}
 \begin{aligned}
   &E[Var(f_i(X_i)|Y_i)]\\
     =&Var[f_i(X_i)]-Var[E(f_i(X_i)|Y_i)]\\
     =&Var[f_i(X_i)]-Var[f_i(\hat{X}_i)].\\
 \end{aligned}
\end{equation}

In the context-aware setting, $Var[f_i(X_i)]$ is a constant, thus the MSE is a function of the variance of each user's estimator. 

%We next consider a basic summation case with $\mathbb{D}=\{0,1\}$  we then  In both cases, 

\section{Privacy-Utility Trade-off}\label{tradeoff}
In this section, we study the privacy-utility tradeoffs by solving the optimization problems defined in Eq. (\ref{utility}). We start by a binary input binary output (BIBO) model where each user has an input/output range of $\{0,1\}$, which well suits the application of survey. In the BIBO model, we illustrate the way of perturbation by analysis on the optimal solutions. we first assume that  $f_i(X_i)=X_i$ due to settings of the first three applications, we discuss optimal solutions for each of the four applications in section \ref{model_applications}.

%After that, we derive the utility-privacy trade-off for the general-model (with multiple input output values). Then we discuss on the optimal output range and compare with the centralized models.

\subsection{Utility-Privacy Tradeoff under the Binary Input Binary Output Model}
The binary model is widely used for survey: each individual's data is first mapped to one bit, than randomly perturbed before publishing to the curator.

For a binary input/output model, each user has a binary input range, \textit{i.e,} $\mathbb{D}=\{0,1\}$ (shown in Fig. \ref{fig:General_Model}). As a direct result, the $Var(X_i)$ in \eqref{largerlaw} becomes $P_1^{i}(1-P_1^{i})$.
Denote the perturbation parameters by:
\begin{equation}
\setlength{\abovedisplayskip}{3pt}
\setlength{\belowdisplayskip}{3pt}
\begin{aligned}
&Pr(Y_i=1|X_i=0)=q^{i}_0, \\
&Pr(Y_i=0|X_i=1)=q^{i}_1.
\end{aligned}    
\end{equation}

Next we derive the concrete optimization objective and constraints. 
By Eq. \eqref{estimator}, the MMSE estimator $\hat{X}_i$ for user $i$ is derived as:
\begin{equation}
\setlength{\abovedisplayskip}{3pt}
\setlength{\belowdisplayskip}{3pt}
\begin{aligned}
    \hat{X}_i=E[X_i|Y_i]
    =&P^{i}_1[\frac{q^i_1}{\lambda^i_0}(1-Y_i)+\frac{1-q^i_1}{\lambda^i_1}Y_i],
\end{aligned}
\end{equation}
where $\lambda^i_0=Pr(Y_i=0)=(1-P_1)(1-q^i_0)+P_1q^i_1$ and $\lambda^i_1=Pr(Y_i=1)=(1-P_1)q^i_0+P_1(1-q^i_1)$.
On the other hand, 
\begin{equation}\label{eq9}
\begin{aligned}
    Var(\hat{X}_i)=&Var\{P^{i}_1[\frac{q^i_1}{\lambda^i_0}(1-Y_i)+\frac{1-q^i_1}{\lambda^i_1}Y_i]\}\\
    =&P^{i}_1\frac{(\lambda^i_0-q_1^{i})^2}{\lambda^i_0\lambda^i_1}.
\end{aligned}
\end{equation}
Take \eqref{eq9} into \eqref{largerlaw}, the each user's MSE function $ \mathcal{E}_i(q^i_0,q^i_1)$ can be derived as
\begin{equation}\label{eq11}
 \mathcal{E}_i(q_0^{i},q_1^{i})=P_1^{i}(1-P_1^{i})-\frac{[P^{i}_1(\lambda^i_0-q_1^{i})]^2}{\lambda^i_0\lambda^i_1}.\\
\end{equation}
For the privacy constraints, by Eq. \eqref{eqlip2}, when the perturbation mechanism satisfies $\epsilon$-LIP: we have:
\begin{equation}\label{eq13}
\setlength{\abovedisplayskip}{3pt}
\setlength{\belowdisplayskip}{3pt}
  e^{-\epsilon}\le{\{F^i_1,F^i_2,F^i_3,F^i_4\}}\le{e^{\epsilon}}, ~~~ \forall{i=1,2...N}.
\end{equation}
where $F^i_1,F^i_2,F^i_3,F^i_4$ are directly derived from Definition (\ref{def:LIP}): $F^i_1(q_0^{i},q_1^{i})=\frac{\lambda^i_0}{q^i_1}$, $F^i_2(q_0^{i},q_1^{i})=\frac{\lambda^i_1}{1-q^i_1}$, $F^i_3(q_0^{i},q_1^{i})=\frac{\lambda^i_0}{1-q^i_0}$, $F^i_4(q_0^{i},q_1^{i})=\frac{\lambda^i_1}{q^i_0}$.
Then, the feasible region  $\mathcal{T}_i$ is defined as those $(q_0^{i},q_1^{i})$ pairs  satisfying constraints in Eq. \eqref{eq13}.

By proposition \ref{thm:decompose}, the optimization problem of Opt-binary-LIP can be reformulated as:
\begin{equation}
\setlength{\abovedisplayskip}{3pt}
\setlength{\belowdisplayskip}{3pt}
\begin{aligned}\label{opt-2}
&\min{\mathcal{E}_i(q^i_0,q^i_1)},\\
 s. t.  &\text{   (\ref{eq13})}, \forall{i=1,2...N}.  
 \end{aligned}
\end{equation}

We have the following result:
\begin{thm}\label{the}
In Opt-binary-LIP, for the $i$-th user, the optimal $(q_0^{i}, q_1^{i})$ pairs that minimize   $\mathcal{E}_i(q_0^{i},q_1^{i})$ in problem (\ref{opt-2}) are: either $q_0^{i*}=P_1^i/e^{\epsilon}$ and $q_1^{i*}=(1-P_1^i)/e^{\epsilon}$, or $q_0^{i*}=1-P_1^i/e^{\epsilon}$ and $q_1^{i*}=1-(1-P_1^i)/e^{\epsilon}$, for any given $\epsilon\geq 0$. The resulting MSE by ($q_0^{i*},q_1^{i*}$) is:
\begin{equation}
\setlength{\abovedisplayskip}{3pt}
\setlength{\belowdisplayskip}{3pt}
\mathcal{E}^*_{bi-LIP}=\sum_{i=1}^{N}\{P_1^{i}(1-P_1^{i})(2e^{-\epsilon}-e^{-2\epsilon})\}.
\end{equation}
\end{thm}
\begin{proof}
Here we outline the proof sketch  (detailed proofs are shown in Appendix. A.

(1) We show that $\mathcal{E}_i$ is monotonically increasing with $q_0^{i}$ and $q_1^{i}$ within the region of $\{q_0^{i}\ge0\}\cap\{q_1^{i}\ge0\}\cap\{q_0^{i}+q_1^{i}\le{1}\}$;

(2)
We simplify the feasible region $\mathcal{T}_i$ by showing that both $\mathcal{E}_i(q_0^{i},q_1^{i})$ and $\mathcal{T}_i$ are symmetric w.r.t. point $(0.5,0.5)$; Then we change $\mathcal{T}_i$ to the monotonic region in step. (1).

(3)
By the monotonicity, showing that the optimal solution is at the boundary of $\mathcal{T}_i$, which is a linear function of $(q_0^{i},q_1^{i})$.

(4)
The final step is to show that optimal solution is at the intersection of two linear functions in step (3) by testing the monotonicity of $\mathcal{E}_i(q_0^{i},q_1^{i})$ on the boundary.
%\end{enumerate}
\end{proof}
Note that, the optimal solution of each user is achieved when $F^i_1(q_0^{i},q_1^{i})=F^i_4(q_0^{i},q_1^{i})=e^{\epsilon}$ (or $F^i_2(q_0^{i},q_1^{i})=F^i_3(q_0^{i},q_1^{i})=e^{\epsilon}$). Intuitively, to increase utility, we need the probability of perturbation as small as possible (when $q_0^{i}+q_1^{i}\le{0.5}$), and the smallest perturbation probability is bounded by the privacy constraints. As a result, the optimal solution is at the point where the privacy requirement is just met. The two optimal $(q_0^{i*},q_1^{i*})$ pairs are symmetric w.r.t. $(0.5, 0.5)$. This is due to the symmetric properties of the binary input/output model. The symmetric properties can also be explained as: if we do not consider privacy, utility is maximized in two ways: the first way is each user publishes his/her data directly; the second way is swapping his/her data from 0 to 1 and 1 to 0 before publishing it. 

From $q_0^{i*}=P_1^i/e^{\epsilon}$ and $q_1^{i*}=(1-P_1^i)/e^{\epsilon}$, we can see that $q_0^{i*}$ is proportional to $P_1^i$, and $q_1^{i*}$ is   proportional to $1-P_1^i$. Intuitively,  from the perspective of one user, when his/her true input value $x_i$'s prior is small, directly revealing $x_i$ will leak too much  information about it. In such cases, to satisfy LIP constraints, a large   perturbation probability is needed to limit the posterior about $x_i$. On the contrary, if  $x_i$ happens with a large prior, directly releasing $x_i$ will reveal little additional information. Thus a small perturbation probability can be used for $x_i$ in this case.

\textbf{Minimizing mean absolute error between the raw data and the perturbed data:} In many applications, the published data by each user should contain certain information even for the parties with no prior knowledge. For example, the reported locations from smart phones: facilities or companies want to take advantages from the aggregated location data frequencies on one hand, the users are uploading locations for location based service on the other. As a result, although the locally published data are perturbed, it should remain a certain level of accuracy. As a result, We want $Y_i$ as close to $X_i$ as possible as long as it still satisfies the privacy constraints. As $P^i$ is given, we want the mean absolute error (MAE) between $X_i$ and $Y_i$ as small as possible, that is, we also want to minimize:
\begin{equation}\label{eqMAE}
E[|X_i-Y_i|], \forall{i\in\{1,2,...,N\}}.
\end{equation}

Regarding at the optimal solutions of the binary model, we have the following proposition:
\begin{prop}\label{propMAE}
The optimal solution that minimize the MSE defined in \eqref{eq11} and the MAE defined in \eqref{eqMAE} simultaneously is: $q_0^{i*}=P_1^i/e^{\epsilon}$ and $q_1^{i*}=(1-P_1^i)/e^{\epsilon}$.
\end{prop}
\begin{proof}

Notice that there are two set of optimal solutions under the Opt-binary-LIP model. 
%When $q_0^{i*}=P_1^i/e^{\epsilon}$ and $q_1^{i*}=(1-P_1^i)/e^{\epsilon}$:

%\begin{equation*}
%\setlength{\abovedisplayskip}{3pt}
%\setlength{\belowdisplayskip}{3pt}
%\begin{aligned}
%\lambda^i_0
%=&P^i_1\frac{(1-P^i_1)}{e^{\epsilon}}+(1-P^i_1)(1-\frac{P^i_1}{e^{\epsilon}})=(1-P^i_1),\\
%\lambda^i_1
%=&P^i_1(1-\frac{(1-P^i_1)}{e^{\epsilon}})+(1-P^i_1)\frac{P^i_1}{e^{\epsilon}}=P^i_1.\\
%\end{aligned}
%\end{equation*}
The MAE defined in $\eqref{eqMAE}$ is
\begin{equation}\label{eqMMAE}
\setlength{\abovedisplayskip}{3pt}
\setlength{\belowdisplayskip}{3pt}
\begin{aligned}
&\sum^1_{x_i=0}\sum^1_{y_i=0}|x_i-y_i|Pr(X_i=x_i)Pr(Y_i=y_i|X_i=x_i)\\
%=&Pr(X_i=1)Pr(Y_i=0|X_i=1)+Pr(X_i=0)Pr(Y_i=1|X_i=0)\\
=&P^i_1q^i_1+(1-P^i_1)q^i_0.
\end{aligned}
\end{equation}
As $P^i_1$ is given, minimizing \eqref{eqMMAE} is equivalent to minimize $q^i_1$ and $q^i_0$, which are also restricted by the privacy constraint, thus the optimal solution $q_0^{i*}=P_1^i/e^{\epsilon}$ and $q_1^{i*}=(1-P_1^i)/e^{\epsilon}$ can also minimize the MAE.
\end{proof}

From the result of proposition. \ref{propMAE}, we can see a trend from the optimal perturbation parameters that minimize the MAE and MSE at the same time, that is when $\epsilon$ increase, the raw data is more likely to be directly published. Notice that, with the optimal solutions,
\begin{equation*}
\setlength{\abovedisplayskip}{3pt}
\setlength{\belowdisplayskip}{3pt}
\begin{aligned}
\lambda^i_0
=&P^i_1\frac{(1-P^i_1)}{e^{\epsilon}}+(1-P^i_1)(1-\frac{P^i_1}{e^{\epsilon}})=(1-P^i_1),\\
\lambda^i_1
=&P^i_1(1-\frac{(1-P^i_1)}{e^{\epsilon}})+(1-P^i_1)\frac{P^i_1}{e^{\epsilon}}=P^i_1.\\
\end{aligned}
\end{equation*}
Thus the distribution of each $X_i$ is identical to $Y_i$, this also guarantees that the output $Y_i$ is very close to  $X_i$.

\textbf{Comparison with Optimized LDP:} We next compare our optimal LIP-based perturbation mechanism with the optimal LDP-based one. Define the Opt-binary-LDP problem to be the same with Opt-binary-LIP in \eqref{opt-2}, except having different privacy constraints of LDP: $\{R_1^i,R_2^i,R_3^i,R_4^i\}\le{e^{\epsilon}}$, where $R_1^i,R_2^i,R_3^i,R_4^i$ are  derived from Definition  (\ref{def:LDP}): $R_1^i=\frac{1-q_1^{i}}{q_0^{i}}$, $R_2^i=\frac{q_0^{i}}{1-q_1^{i}}$, $R_3^i=\frac{q_1^{i}}{1-q^i_0}$ and $R_4^i=\frac{1-q_0^{i}}{q_1^{i}}$.
\begin{prop}
In Opt-binary-LDP, for the $i$-th user, the optimal solution for $(q^i_0,q^i_1)$ is $q_0^{i*}=q_1^{i*}=\frac{1}{e^{\epsilon}+1}$, which results in a MSE $\mathcal{E}^*_{bi-LDP}$ of:
\begin{equation}
\sum_{i=1}^{N}\{P_1^{i}(1-P_1^{i})-\frac{[P_1^{i}(1-P_1^{i})(1-e^{\epsilon})]^2}{(1-P_1^{i}+P_1^{i}e^{\epsilon})(e^{\epsilon}-P_1^{i}e^{\epsilon}+P_1^{i})}\}.
\end{equation}
Given any fixed $\epsilon\ge0$, we have $\mathcal{E}^*_{bi-LDP}\ge{\mathcal{E}^*_{bi-LIP}}$.
\end{prop}

\begin{proof}
The proof of the optimal solution is similar to that of Opt-binary-LIP, the only difference is the feasible region in Opt-binary-LIP is flexible for different priors, while the feasible region in Opt-binary-LDP is fixed. The optimal solution also coincides with the one used in \cite{Tianhao}.

For the second part, it's easy to check by taking derivative over $e^{\epsilon}$ that $\mathcal{E}^*_{bi-LIP}\le{\mathcal{E}^*_{bi-LDP}}$, where $\mathcal{E}^*_{bi-LIP}={\mathcal{E}^*_{bi-LDP}}$ if $\epsilon=0$ or $\epsilon=\infty$. 
This means LIP provides increased utility given any $\epsilon$. We then taking derivative of $P^i_1$ over $\Delta\mathcal{E}^*=\mathcal{E}^*_{bi-LDP}-\mathcal{E}^*_{bi-LIP}$, result shows that $\frac{\partial\Delta\mathcal{E}^*}{\partial{P^i_1}}=0$ when $P^i_1=0.5$. As $\Delta\mathcal{E}^*_{(P^i_1=0.5)}\ge{0}$, $\mathcal{E}^*_{bi-LDP}\ge\mathcal{E}^*_{bi-LIP}$ for any $P^i_1$. Result also shows that as $|P^i_1-0.5|$ grows, $\Delta\mathcal{E}^*$ also increases. 
\end{proof}
The above result shows that,  Opt-binary-LIP always achieves a better utility than Opt-binary-LDP under any $\epsilon$, this is because explicitly considering prior in the privacy definition allows a larger search space for optimal parameters than that in Opt-binary-LDP. We can also learn this result from the aspect of information theory, when $|P^i_1-0.5|$ is small, then the $H(X_i)$ is large, which means $X_i$ has the largest amount of uncertainty, thus knowing the prior of $X_i$ does not help in perturbation; However, when $H(X_i)$ is small, prior knowledge of $X_i$ provides a clearer indication of the real value, thus the context-aware model achieves enhanced advantages.

The binary case is an illustrative example that shows how we derive the optimal solutions, as the problem is equivalent to finding the maximum in a monotonically increasing region, thus the methods with goal to find the station points in the feasible region such as Lagrangian multiplier and KKT conditions are not applicable. On the other hand, thanks to the theorem that the maximum value of the monotocity function can be found at the boundary, the problem can be first regarded as developing the monotocity region and then finding the boundary values of the parameters. 
Moreover, with the goal to minimize the MAE, we can further restrict the optimal solutions to make the model more practical.

 When each user has a same prior and when the perturbation channel is symmetric, these can be considered as special cases of the BIBO model, and details are referred to our conference version in \cite{Jian1805:Context}.
\subsection{Utility-Privacy Tradeoff under MIMO Model}

More generally, we study the case in which $\mathbb{D}$ has a large domain, \textit{i.e.}, $\mathbb{D}=\{a_1, a_2,...,a_d\}$ with prior distribution $Pr(X_i=a_m)=P^i_m$. This case well suits the application of summation and weighted summation and we denote this model as MIMO model. The optimization of the MIMO model is obviously obscure, as there are $d^2$ parameters that need to optimize for a single user.  (shown in \ref{fig:MIMO_Model})
% \begin{figure}[htp]
%\centering
%\includegraphics[width=5cm]{MIMO.eps}
%\caption{Perturbation mechanism with multiple input-multiple output}
%\centering
%\label{fig:example_variance}
%\end{figure}

Assume that a random-response perturbation mechanism which satisfies $\epsilon$-LIP takes input $X_i$ and output $Y_i$ ($Y_i$ has the same range, we discuss the optimal output range after deriving the main results) with probability $Pr(Y_i=a_k|X_i=a_m)=q^i_{mk}$. Denote $Pr(Y_i=a_k)=\lambda^i_k$  . Thus, by \eqref{eqlip2}:
\begin{equation}\label{multi-privacy}
    e^{-\epsilon}\le\frac{Pr(X_i=a_m)}{Pr(X_i=a_m|Y_i=a_k)}\le{e^{\epsilon}}, \forall{m,k\in{1,2,...,d}} 
\end{equation}
By Bayes rules, \eqref{multi-privacy} can be transferred to:
\begin{equation}
    e^{-\epsilon}\le\frac{\lambda^i_k}{q^i_{mk}}\le{e^{\epsilon}}, \forall{m,k\in{1,2,...,d}} 
\end{equation}
In \eqref{largerlaw}, $Var[X_i]=\sum^d_{m=1}a_m^2P^i_m-(\sum^d_{m=1}a_mP^i_m)^2$, and
\begin{equation}\label{Multiconstraint}
\begin{aligned}
   \hat{X_i}=&E[X_i|Y_i]
   =\sum^d_{m=1}a_jPr(X_i=a_m|Y_i)\\
   =&\sum^d_{m=1}\sum^d_{k=1}a_mPr(X_i=a_m|Y_i=a_k)\mathbbm{1}^i_{k},
\end{aligned}
\end{equation}
where $\mathbbm{1}^i_{k}$ is the indicator function of $\mathbbm{1}^i_{\{Y_i=a_k\}}$, which is 1 if $Y_i=a_k$ and 0 if not, thus $\mathbbm{1}^i_{k}$ can be regarded as a binary random variable which has the distribution of: $Pr(\mathbbm{1}^i_{k}=1)=\lambda^i_{k}$ and $Pr(\mathbbm{1}^i_{k}=0)=1-\lambda^i_{k}$, as a result: $Var[\mathbbm{1}^i_k]=\lambda^i_k(1-\lambda^i_k)$ and $Cov[\mathbbm{1}^i_k,\mathbbm{1}^i_l]=-\lambda^i_k\lambda^i_l$.

\begin{small}

\begin{equation}
\setlength{\abovedisplayskip}{3pt}
\setlength{\belowdisplayskip}{3pt}
\begin{aligned}
   &Var[\hat{X_i}]
   =\sum^d_{m=1}\sum^d_{n=1}\sum^d_{k=1}a_ma_nq^i_{mk}q^i_{nk}Var[\mathbbm{1}^i_{k}]\\
 +&\sum^d_{m=1}\sum^d_{n=1}\sum^d_{k=1}\sum^d_{l=1;l\neq{k}}a_ma_nq^i_{mk}q^i_{nl}Cov[\mathbbm{1}^i_{k},\mathbbm{1}^i_{l}]\\
   %=&\sum^d_{m=1}\sum^d_{n=1}\sum^d_{k=1}a_ma_nP^i_mP^i_n\frac{q^i_{mk}q^i_{nk}(1-\lambda^i_k)}{\lambda^i_k}\\-&\sum^d_{m=1}\sum^d_{n=1}\sum^d_{k=1}\sum^d_{l=1;l\neq{k}}a_ma_nP^i_mP^i_nq^i_{mk}q^i_{nl}\\
   =&\sum^d_{m=1}\sum^d_{n=1}\sum^d_{k=1}a_ma_nP^i_mP^i_nq^i_{mk}(\frac{q^i_{nk}(1-\lambda^i_{k})}{\lambda^i_k}-\sum^d_{l=1;l\neq{k}}q^i_{nl})\\
   =&\sum^d_{m=1}\sum^d_{n=1}\sum^d_{k=1}a_ma_nP^i_mP^i_nq^i_{mk}(\frac{q^i_{nk}}{\lambda^i_k}-1).
\end{aligned}
\end{equation}
\end{small}

Denote $\mathbf{q^i}$ as the set of $\{q^i_{11},q^i_{12},...,q^i_{1d},...,q^i_{dd}\}$, By proposition. \ref{thm:decompose}. The optimization problem of the MIMO model, Opt-mino-LIP is formulated as:
\begin{equation}\label{multi-opt}
\setlength{\abovedisplayskip}{3pt}
\setlength{\belowdisplayskip}{3pt}
\begin{aligned}
            &\min{\mathcal{E}_i(\mathbf{q^i})}, \text{\eqref{eqMAE}}\\
 s. t.  &\text{   (\ref{multi-privacy})},  \forall{i=1,2...N}.  
 \end{aligned}
\end{equation}
\begin{thm}\label{thm3}
For the constraint optimization problem defined in \eqref{multi-opt}. The optimal solutions is: $q^{i*}_{mm}=1-(1-P^i_m)/e^{\epsilon}$, $q^{i*}_{mk}=P^i_k/e^{\epsilon}$, for all $m,k\in\{1,2,...,d\}$, $m\neq{k}$.
\end{thm}
Brief steps of proof (detailed proof is shown in Appendix. B.
\begin{itemize}
    \item Regardless of the privacy constraints, we first show that the maximum value of $Var[\hat{X_i}]$ can be reached when the $q^i_{mm}=1$, for $m\in\{1,2,...,d\}$. also, the minimum value of $Var[\hat{X_i}]$ can be reached when $q^i_{mk}=\lambda_k$, for any $m,k\in\{1,2,...,d\}$. 
    \item We then take derivative over a randomly chosen parameter: $q^i_{mk}$ and derive the monotonicity region.
    \item In each monotonicity region we show that the parameters that decrease will first reach the boundary. Thus the optimal solutions lies at the lower boundaries. As we can permute the sequence of $Y_i$, we can find $d!$ optimal solutions.
    \item Considering minimizing the MAE defined in \eqref{eqMAE}, there is only one set of optimal solution remaining.
\end{itemize}

From Theorem \ref{thm3}, we can see that the optimal solutions of the Opt-mimo-LIP also lies at the boundary of the privacy constraints. As $\epsilon$ increases, $\forall{m\in\{1,2,...,d\}}$, all the $q^i_{mm}$s are increasing while all the $q^i_{mk}$s are decreasing ($m\neq{k}$), and the value of $q^i_{mk}$s are proportional to $P^i_k$s. thus an input value that has a larger prior should also be output with a large probability.

For example, consider a model in which  $\mathbb{D}=\{1,2,3\}$ the prior of one of the users is given as: $P_1=0.1$, $P_2=0.2$, $P_3=0.7$. By Theorem \ref{thm3} $q^*_{11}=1-0.9/\epsilon$, $q^*_{22}=1-0.8/\epsilon$, $q^*_{33}=1-0.3/\epsilon$, $q^*_{21}=q^*_{31}=0.1/\epsilon$, $q^*_{12}=q^*_{32}=0.2/\epsilon$, $q^*_{13}=q^*_{23}=0.7/\epsilon$. When $\epsilon$ increases, each value is more likely to be directly published. As 3 has a larger prior than 1 and 2, if 1 and 2 are not directly published, they are more likely to be published as a 3.

Notice the main goal of \eqref{multi-opt} is to minimize the MSE rather than the MAE, so, we still deploy the MMSE estimator. Regardless of the goal to minimize MAE, there are $d!$ optimal solutions for a single user in the problem of Opt-mimo-LIP, as we can randomly permute the order of $Y_i$. When minimizing the MAE, there is a unique solution remaining, as : $\lambda^i_k=P^i_k(1-(1-P^i_k))/e^{\epsilon}+\sum^d_{j\neq{k}}P^i_jP^i_k/e^{\epsilon}=P^i_k$.

Similarly, we formulate the optimization problem for the LDP, which is the same goal defined in \eqref{multi-opt} while subject to the constraints of LDP: $\forall m,n,k\in\{1,2,...,d\}$, $m\neq{n}$, there is $\frac{q^i_{mk}}{q^i_{nk}}\le{e^{\epsilon}}$. We derive the optimal solutions for Opt-mimo-LDP as $q^{i*}_{mm}=\frac{e^{\epsilon}}{e^{\epsilon}+d-1}$, $q^{i*}_{mk}=\frac{1}{e^{\epsilon}+d-1}$, $\forall{m,k\in{1,2,...,d}}$, $m\neq{k}$. Direct comparison between the two mechanisms involves large amount of calculation, thus we compare them in simulation . 

\textbf{Optimal Output Range:} In terms of the optimal range for each $Y_i$, previous models are considering $Y_i$ has the same domain with $X_i$, now we consider that the outputs take values from a different domain size. Denote the domain of the input as $\mathbb{D}_x=\{a_1,a_2,...,a_d\}$; the domain of  the output as $\mathbb{D}_y=\{a_1,a_2,...,a_f\}$. When $d$ is fixed, we want to find the optimal value of $f$.

\begin{thm}\label{thm:samerange}
In the Opt-mimo-LIP problem, when the input range of $d$ is fixed, the optimal output range $f^*$ is $f^*=d$.
\end{thm}
Detailed proof is shown in Appendix. C.

From Theorem 4, we know the optimal output range is the same with the input. This property is helpful for model setting. 

\subsection{Applications to Histogram Estimation}\label{model_applications}

Obviously, for the application of survey, the Opt-binary-LIP is suitable, and Opt-mimo-LIP can deal with the problem of direct summation and weighted summation. We now consider the  problem of estimating a histogram.

The difference between the application of the histogram and other applications is even though for each user, before uploading data, the perturbation mechanism still takes one input data and outputs one data, the estimator of a histogram is a random vector rather than a random variable, as the value of each data stands for a category.

Derive the optimal estimator of the histogram vector, $\hat{\mathbb{S}}=\{\hat{S_1},\hat{S_2},...,\hat{S_d}\}=\{E[S_1|\bar{Y}],E[S_2|\bar{Y}],...,E[S_d|\bar{Y}]\}$, with each entry:
\begin{equation}
\setlength{\abovedisplayskip}{3pt}
\setlength{\belowdisplayskip}{3pt}
\begin{aligned}
    E[S_k|\bar{Y}]=E[\sum^N_{i=1}\mathbbm{1}_{\{X_i=a_k\}}|\bar{Y}]
    =\sum^N_{i=1}E[\mathbbm{1}_{\{X_i=a_k\}}|Y_i].
\end{aligned}
\end{equation}
%The second norm distance of $\mathbb{S}$ and $\hat{\mathbb{S}}$ is 
%\begin{equation}
%    \norm{\mathbb{S}-\hat{\mathbb{S}}}^2=\sum^d_{k=1}(S_k-\hat{S}_k)^2.
%\end{equation}
Thus the mean square error of the estimation is
\begin{small}
\begin{equation}\label{Error}
\begin{aligned}
    &E[\sum^d_{k=1}(S_k-\hat{S}_k)^2]\\
    =&\sum^d_{k=1}E[(\sum^N_{i=1}\{\mathbbm{1}_{\{X_i=a_k\}}-E[\mathbbm{1}_{\{X_i=a_k\}}|Y_i]\})^2]\\
    \overset{a}{=}&\sum^d_{k=1}\{\sum^N_{i=1}E[(\{\mathbbm{1}_{\{X_i=a_k\}}-E[\mathbbm{1}_{\{X_i=a_k\}}|Y_i]\})^2]\\
    =&\sum^N_{i=1}\sum^d_{k=1}\{Var(\mathbbm{1}_{\{X_i=a_k\}})-Var(E[\mathbbm{1}_{\{X_i=a_k\}}|Y_i])\}.\\
\end{aligned}
\end{equation}
\end{small}
The (a) of \eqref{Error} is because each user's local error is independent, and the expectation of the unbiased estimator is identical to that of the estimated value.

Thus the problem of estimating a histogram can be formulated as:
\begin{equation}\label{eq:histogram}
\setlength{\abovedisplayskip}{3pt}
\setlength{\belowdisplayskip}{3pt}
\begin{aligned}
&\min\eqref{Error}, \eqref{eqMAE}\\
 s. t.  &\text{   (\ref{Multiconstraint})},  \forall{i=1,2...N}.  
 \end{aligned}
\end{equation}

\begin{thm}\label{mimo-bianry}
The optimal perturbation parameters of problem defined in \eqref{eq:histogram} is $q^{i*}_{mm}=1-(1-P^i_m)/e^{\epsilon}$, $q^{i*}_{mk}=P^i_k/e^{\epsilon}$, for all $m,k\in\{1,2,...,d\}$, $m\neq{k}$.
\end{thm}
%\begin{proof}
%$\forall{i\in{\{1,2,...,N\}}}$, if the perturbation mechanism $\mathcal{M}$ of user $i$ satisfies $\epsilon$-Opt-mino-LIP, then for any $X_i=a_m\in\mathbb{D}$, there is $Pr(Y_i=a_m|X_i=a_m)^*=1-\frac{1-P^i_m}{e^{\epsilon}}$ and $Pr(Y_i=a_m|X_i=a_k)^*=P^i_k/e^{\epsilon}$, for all $m,k\in\{1,2,...,d\}$, $m\neq{k}$. Then:
%\begin{equation}
%\sum^{d}_{j=1,j\neq{m}}q^{i*}_{mj}=\frac{1}{e^{\epsilon}}\sum^{d}_{j=1,j\neq{m}}P^i_j=\frac{1-P^i_m}{e^{\epsilon}}.
%\end{equation}
%In the problem of $\epsilon$-Opt-binary-LIP, if the interesting secret is $a_m$, then $\forall{i\in{\{1,2,...,N\}}}$, $f_i(X_i)=\mathbbm{1}_{\{X_i=a_m\}}$. Then the optimal solution is $Pr(Y_i=1|f_i(X_i)=0)^*=P^i_1/e^{\epsilon}$ and $Pr(Y_i=1|f_i(X_i)=1)^*=1-\frac{1-P^i_1}{e^{\epsilon}}$. We know that $Pr(f_i(X_i=1))=Pr(X_i=a_m)=P^i_m$, thus these two metrics are equivalent.
%\end{proof}
\begin{proof}
Detailed proof is shown in Appendix. D.
\end{proof}

The application of histogram can be viewed as a special case of the Opt-mimo-LIP because in the Opt-mimo-LIP, the optimal solutions are to make the estimation of each user 's data as close to the real value as possible, while the request of the histogram is to classify each user's data as accurate as possible, when the perturbation parameters are making each of the estimation as accurate as possible, it also handle the problem of identifying its category.

For implementation of the histogram, users are still publishing data according to the Opt-mimo-LIP perturbation mechanism. After receiving the published data by each user, the curator then estimates $\{E[S_1|Y_i],E[S_2|Y_i],...,E[S_d|Y_i]\}=\{Pr(X_i=a_1|Y_i), Pr(X_i=a_2|Y_i),..., Pr(X_i=a_d|Y_i)\}$ according to the optimal perturbation parameters.

\subsection{Centralized Information Privacy (Binary case)}\label{sec:centralized}

For comparison purposes, we now derive the formula for centralized information privacy under the binary perturbation mechanism, we first illustrate the idea by the global prior model:
Consider N users are in a data aggregating model who directly submit their data to a trusted central aggregator, who publishes a perturbed version $Y$ of $S=f(X)$ using a CIP mechanism $\mathcal{M}(S)$. See Fig. \ref{fig:Central}. Assume that each user's input value is taken as a random variable $X_i$ which takes value from a range of $\mathbb{D}_i=\{0,1\}$, $\forall{i\in\{1,2,...,N\}}$. It is known that the prior distribution of $X_i$ is: $P^i_1=Pr(X=1)$, $\forall{i\in\{1,2,...,N\}}$. Denote $\mathcal{D}$ as the database holding all users' data, where $\mathcal{D}_i=X_i$. Thus $Pr(\mathcal{D}_i=1)=P_1$. 
%\begin{figure}[htp]
%\centering
%\includegraphics[width=6cm]{Centralized_2.eps}
%\centering
%\caption{Centralized Binary Information Privacy Model}
%\label{fig:Central}
%\end{figure}
Assume that the data the curator wants to aggregate is $S=\sum^{N}_{i=1}\mathcal{D}_i$ and the trusted third party perturbs $S$ by using a randomized response mechanism $\mathcal{M}$.
Notice that $\mathcal{D}_1,\mathcal{D}_2,...,\mathcal{D}_N$ is a N times Bernoulli trial, Thus:
\begin{equation}
\setlength{\abovedisplayskip}{3pt}
\setlength{\belowdisplayskip}{3pt}
    Pr(S=s)=C^s_NP^s_1(1-P_1)^{N-s},
\end{equation}
\begin{defi}\label{defi_central}
A mechanism $\mathcal{M}$ which takes input values of $\mathcal{D}=\{X_1,X_2,...,X_N\}$ and output $Y$ satisfies centralized information privacy if $\forall{i,y\{1,2,...,N\}}$ :
\begin{equation}\label{def:Central}
\setlength{\abovedisplayskip}{3pt}
\setlength{\belowdisplayskip}{3pt}
    e^{-\epsilon}\le{\frac{Pr(\mathcal{D}_i)}{Pr(\mathcal{D}_i|Y=y)}}\le{e^{\epsilon}},
\end{equation}
\end{defi}
Intuitively, the difference between the centralized and the localized IP is the input of the mechanism: the input of the centralized IP is a database, while the input of the LIP is each user's local data. In terms of the implementation, in the settings of CIP, all the raw data is hold by the trusted server, and the mechanism is run at the server rather than at each user.  The configuration of the centralized IP model is shown in Fig.\ref{fig:Central}.

Notice that the right side of the configuration of the centralized IP is multiple input multiple output LIP for a single user. Based on this observation, we have the following theorem.
\begin{thm}\label{thm:central}
If the mechanism $\mathcal{M}$ which takes the inputs of $S$ and output $Y$ satisfies $\epsilon$-LIP, then it also satisfies the centralized model defined in Definition (\ref{defi_central}).
\end{thm}
\begin{proof}
If the mechanism $\mathcal{M}$ $\epsilon$-LIP:
\begin{equation*}   e^{\epsilon}Pr(S=s)\le{Pr(S=s|Y=y)}\le{e^{\epsilon}Pr(S=s)},
\end{equation*}
then, for any $j\in{0,1}$
\begin{small}
\begin{equation}
\setlength{\abovedisplayskip}{3pt}
\setlength{\belowdisplayskip}{3pt}
\begin{aligned}
    Pr(\mathcal{D}_i=j|Y=y)=&\sum^{N}_{s=1}Pr(\mathcal{D}_i=j|S=s)Pr(S=s|Y=y)\\
               \le&\sum^{N}_{s=1}Pr(\mathcal{D}_i=j|S=s)Pr(S=s)e^{\epsilon}\\
               =&e^{\epsilon}Pr(\mathcal{D}_i=j).\\
\end{aligned}
\end{equation}
\end{small}
Do this again for the other side we can get:
\begin{equation}
\setlength{\abovedisplayskip}{3pt}
\setlength{\belowdisplayskip}{3pt}
    Pr(\mathcal{D}_i=j|Y=y)\ge{e^{-\epsilon}Pr(\mathcal{D}_i=j)}.
\end{equation}
Thus, if the requirement of  localized information privacy is met, then centralized information privacy is also satisfied. It's easy to check it's not vise versa as $D_i$, $Y_i$ and $S$ does not form a Markov chain.
\end{proof}
From Theorem \ref{thm:central}, we know that localized information privacy infers centralized information privacy.
In terms of the utility, we still want to minimize the mean square error: $E[(S-\hat{S})^2]$ using the MMSE estimator: $\hat{S}=E[S|Y]$. Thus the optimization problem with goal to minimize the MSE while subject to the privacy constraints satisfying Definition (5) involves large amount of calculation and is not the goal of this paper, however, we can still find the optimal utility-privacy trade off without specifying each parameter.

\begin{prop}\label{centralized}
When there is a global prior for each user, the utility-privacy trade off of the centralized Information privacy is found at the solutions of the optimization problem that
\begin{equation}
\begin{aligned}
&\mathcal{E}_{CIP}=E[(S-\hat{S})^2],\\
s.t.\quad e^{-\epsilon}&NP_1\le{\hat{S}}\le{N-e^{-\epsilon}(N-P_1N)}.
\end{aligned}
\end{equation}
\end{prop}
\begin{proof}
We know by previous results:
\begin{equation}
    E[(S-\hat{S})^2]=Var(S)-Var(\hat{S}),
\end{equation}
where $Var(S)=NP_1(1-P_1)$. We now derive  the $Var(\hat{S})$.
Base on the privacy constraints, we further derive the privacy metric:
\begin{equation}
\begin{aligned}
    \frac{Pr(\mathcal{D}_i=1)}{Pr(\mathcal{D}_i=1|Y)}
    =&\frac{P_1}{\sum^N_{s=1}Pr(\mathcal{D}_i=1|S=s)Pr(S=s|Y)}\\
    %=&\frac{NP_1}{\sum^N_{s=1}sPr(S=s|Y)}\\
    =&\frac{NP_1}{E[S|Y]},\\
\end{aligned}
\end{equation}
On the other hand:
\begin{equation}
\begin{aligned}
    \frac{Pr(\mathcal{D}_i=0)}{Pr(\mathcal{D}_i=0|Y)}
    =&\frac{1-P_1}{\sum^N_{s=1}Pr(\mathcal{D}_i=0|S=s)Pr(S=s|Y)}\\
    %=&\frac{N-NP_1}{\sum^N_{s=1}(N-s)Pr(S=s|Y)}\\
    =&\frac{N-NP_1}{N-E[S|Y]}.\\
\end{aligned}
\end{equation}
So, we have $e^{-\epsilon}\le\{\frac{NP_1}{E[S|Y]},\frac{NP_1}{N-E[S|Y]}\}\le{e^{\epsilon}}$, which further infers $e^{-\epsilon}NP_1\le{\hat{S}}\le{N-e^{-\epsilon}(N-P_1N)}$. 
%We know that the distribution of $S$ is binomial, as an unbiased estimator, the distribution of $\hat{S}$ is also binomial. The difference between $S$ and $\hat{S}$ comes from the range. while $S$ has the range of $N$, $\hat{S}$ has the range from $e^{-\epsilon}NP_1$ to ${N-e^{-\epsilon}P_1N}$, which is a shrink from 0 to N with a factor of $(1-e^{-\epsilon})$ when $\epsilon=0$, the range is only $NP_1$, when $\epsilon$ goes infinity, the range is from 0 to $N$. As a result:
%\begin{equation}
%\begin{aligned}
 %   E[(S-\hat{S})^2]=&Var(S)-Var(\hat{S})\\
%    =&NP_1(1-P_1)-N(1-e^{-\epsilon})P_1(1-P_1)\\
 %   =&NP_1(1-P_1)e^{-\epsilon}.
%\end{aligned}
%\end{equation}
\end{proof}

To deriving the optimal solutions of each parameters is not a main objective of this paper, however, we want to show the comparison between the centralized IP and the localized IP.
%Now, we use theorem \ref{centralized} to derive the result of the model in which each user has a local prior (each user's prior is different with each other). 

%Assume that for user $i$, where $i\in\{1,2,...,N\}$, $Pr(D_i=1)=P^i_1$.

%\begin{thm}
%When each user has a local prior, the minimized mean square error using the mechanism $M$ which satisfies the privacy constraints defined in definition 1. is $\sum^{N}_{i=1}P^i_1(1-P^i_1)e^{-\epsilon}$.
%\end{thm}
%\begin{proof}
%Notice that each user is still assumed to be independent with each other, thus the distribution of $S$ is still equivalent to a $N$ times Bernoulli trail, but with different $P^i_1$.
%Thus the overall variance becomes $Var(S)=\sum^{N}_{i=1}P^i_1(1-P^i_1)$.
%Notice that the distribution of the estimator is also Bernoulli and also depends on the privacy budget $\epsilon$, which means $e^{-\epsilon}E(S)\le{\hat{S}}\le{N-e^{-\epsilon}(N-E(S))}$, where $E(S)=\sum^{N}_{i=1}P^i_1$. As a result, $Var(\hat{S})$ is also decreased than $Var(S)$ by a factor of $(1-e^{\epsilon})$
%\end{proof}

\section{Evaluation}\label{sec:sim}
In this section, we use simulation to validate our analytical results. In the first part, we validate our analysis using synthetic data and via Monte-Carlo simulation. We first show  the advantages of the context-aware privacy notion (based on LIP) versus the context-free notion (based on LDP),   by comparing their utility-privacy tradeoffs. Then we compare different models of LIP and LDP under different privacy budgets and prior distributions in both binary case and MIMO case (with both global priors and local priors). At last, we validate the analytical results through Monte-Carlo simulation, and compare utility-privacy trade offs provided by the localized models and the centralized model In section \ref{sec:centralized}. In the second part, we evaluate on two real-world datasets:  Karosak (click-streams of websites) and Gowalla (location check-in data).

%our objective is to illustrate that  Opt-LIP model achieves utility gain in two-fold: firstly, it allows formulating the estimator with prior as a parameter which  connects utility directly with prior knowledge; Secondly, the notion of LIP explicitly contains prior, as a result, as the privacy constraints in the the problem of maximizing utility, it provides more feasible solutions than prior-free notions such as LIP.

We evaluate utility by  square root average MSE $(\sqrt{\mathcal{E}/N})$. This is because MSE depends on the number of users, for comparison purposes,  we first average the MSE to normalize the influence of user count. In addition, since MSE is square error, we take the square root in order to make it comparable with absolute error, which roughly means how much deviation percentage is the result from the exact value. Note that, doing so does not change the  optimal solutions in any of our optimization problems. In addition, since LIP achieves a relaxed privacy level than LDP, it  is difficult to  compare their utilities under the same privacy level. Thus,   we will compare their optimal utilities under any given privacy budget $\epsilon$.
All the simulations are done in Matlab (R2016a) on a Dell desktop (OptiPlex 7040; CPU: Intel (R) Core (TM) i5-6500 @ 3.2GHz; RAM 8.0 GB; OS: windows 64bit).

\subsection{Simulation Results on Synthetic Data}

\subsubsection{Benefit of   Context-awareness}
First we would like to compare the utility-privacy tradeoffs between $\epsilon$-LIP and $\epsilon$-LDP using the binary model, and the goal is to observe the advantage of  our proposed context-aware notion versus context-free notion of LDP. Intuitively, the utility gain of the former can be attributed to two factors: 1) using the prior in the MMSE estimator, which improves the accuracy compared with estimators that do not use prior knowledge; 2) the privacy guarantee of LIP is relaxed compared with LDP, by explicitly modeling prior in the definition. As a result, less perturbation is needed to satisfy the same privacy budget $\epsilon$. The latter factor is already proven in proposition 3. To decouple the influence of the above two factors, we compare  the utility-privacy tradeoff of Opt-LIP with two other schemes: Opt-LDP (defined earlier), as well as 
context-free  LDP adopted in previous work \cite{Tianhao}. The prior-unaware estimator used in \cite{Tianhao} is denoted as $\hat{C}$, which treats $X_i$ as instances rather than variables:
\begin{equation}
\setlength{\abovedisplayskip}{3pt}
\setlength{\belowdisplayskip}{3pt}
\hat{C}=\frac{\sum^N_{i=1}Y_i-N{p_i}}{1-2p_i},
\end{equation}
where $p_i=\frac{e^{\epsilon}}{e^{\epsilon}+1}$ as we discussed in section \ref{tradeoff}.
This is an unbiased estimator under the binary symmetric  channel (BSC) model. The MSE function of this estimator is:
\begin{equation}
\setlength{\abovedisplayskip}{3pt}
\setlength{\belowdisplayskip}{3pt}
\begin{aligned}
E[(S-\hat{C})^2]=Var[\hat{C}]
=\frac{Np_i(1-p_i)}{(1-2p_i)^2}
\end{aligned}
\end{equation}
\begin{figure}[t] 
\centering 
\includegraphics[width=9cm]{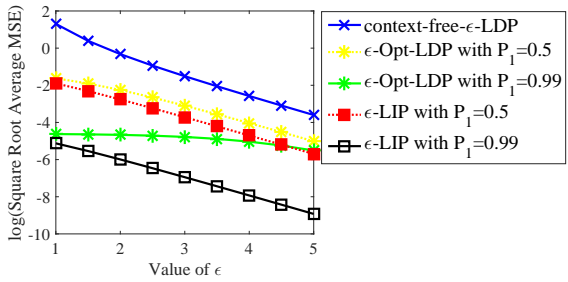} 
\caption{The utility-privacy tradeoff comparison between prior-aware and prior-free models (log scale for y-axis)} 
\label{fig:prior_compare} 
\end{figure}

The comparison is shown in Fig. \ref{fig:prior_compare}. The privacy budget $\epsilon$ changes from $1$ to $5$  with a step of $0.5$.  For now we assume that $N$ users share the same global prior.  We can see that, the square root average MSE of ``$\epsilon$-Opt-LDP" is always smaller than that of  ``context-free LDP" under any given $\epsilon$. When $P_1=0.5$ (prior is uniformly distributed), the distance between these two models is smaller; when the prior is more skewed, advantage of the former is even enhanced. This validates the benefit of  the prior-aware estimator. On the other hand, by comparing the curves of ``$\epsilon$-Opt-LDP" and ``$\epsilon$-Opt-LIP" (using the same MMSE estimator),  the error of Opt-LIP is always smaller than that of Opt-LDP, and the gap between the two models increases when $P_1=0.99$. This result confirms that our relaxed prior-aware privacy notion leads to increased utility.
%Since the context-free LDP provides the worst utility, we only focus on prior-aware notions and estimators in the following.
 When $P_1=0.99$, users' inputs are highly certain, merely considering prior in the estimator can already result  in     accurate aggregation. Thus advantage of  the context-aware 
$\epsilon$-LIP  is even enhanced.

\begin{figure}[ht]
\centering 
\subfigure[Utility-privacy tradeoff comparison with 500 users with each $|\mathbb{D}|=5$] 
{ \includegraphics[width=4cm]{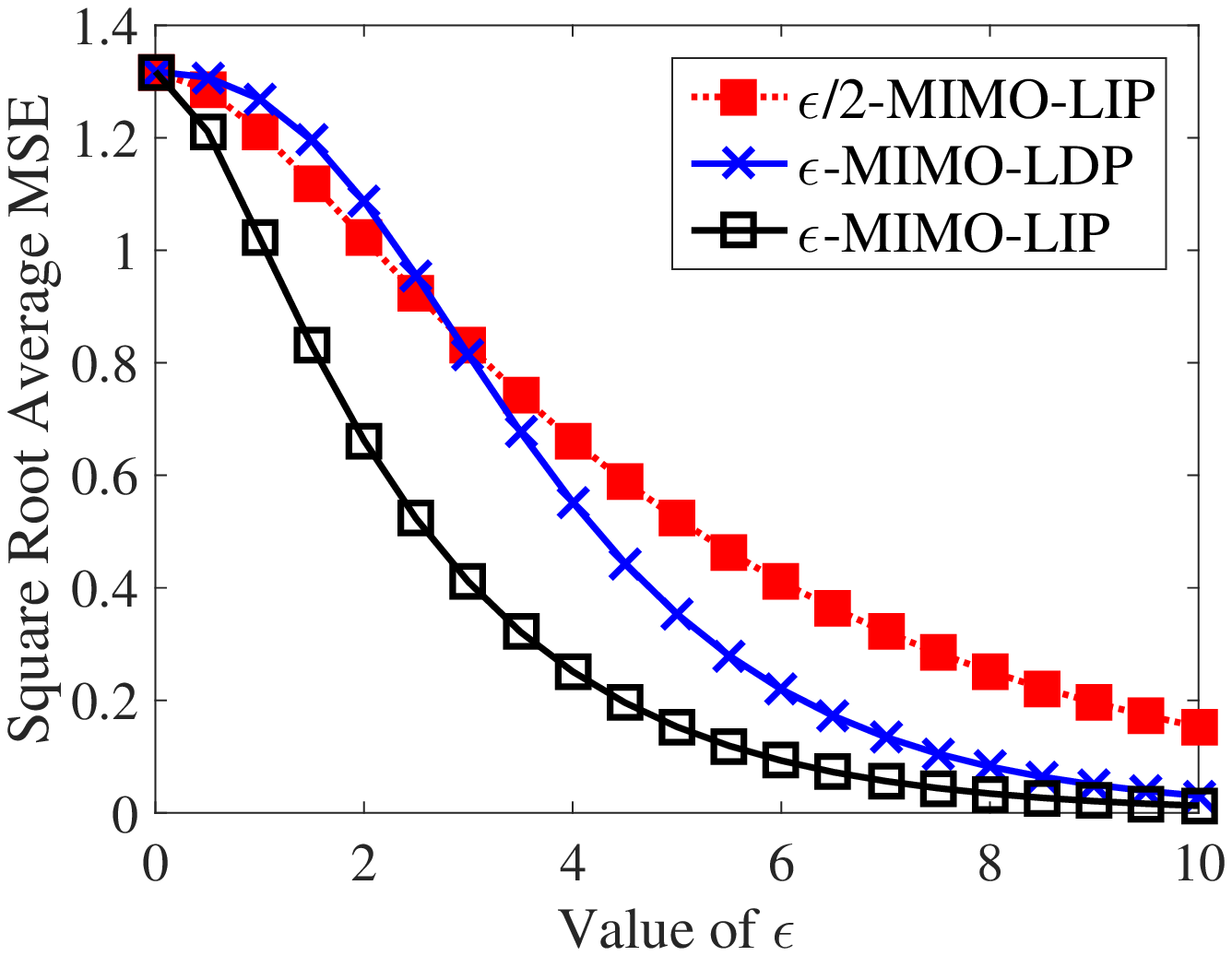} 
\label{dfixed} } 
\subfigure[Utility-privacy tradeoff comparison with $|\mathbb{D}|$ increasing from 2 to 20] 
{ \includegraphics[width=4cm]{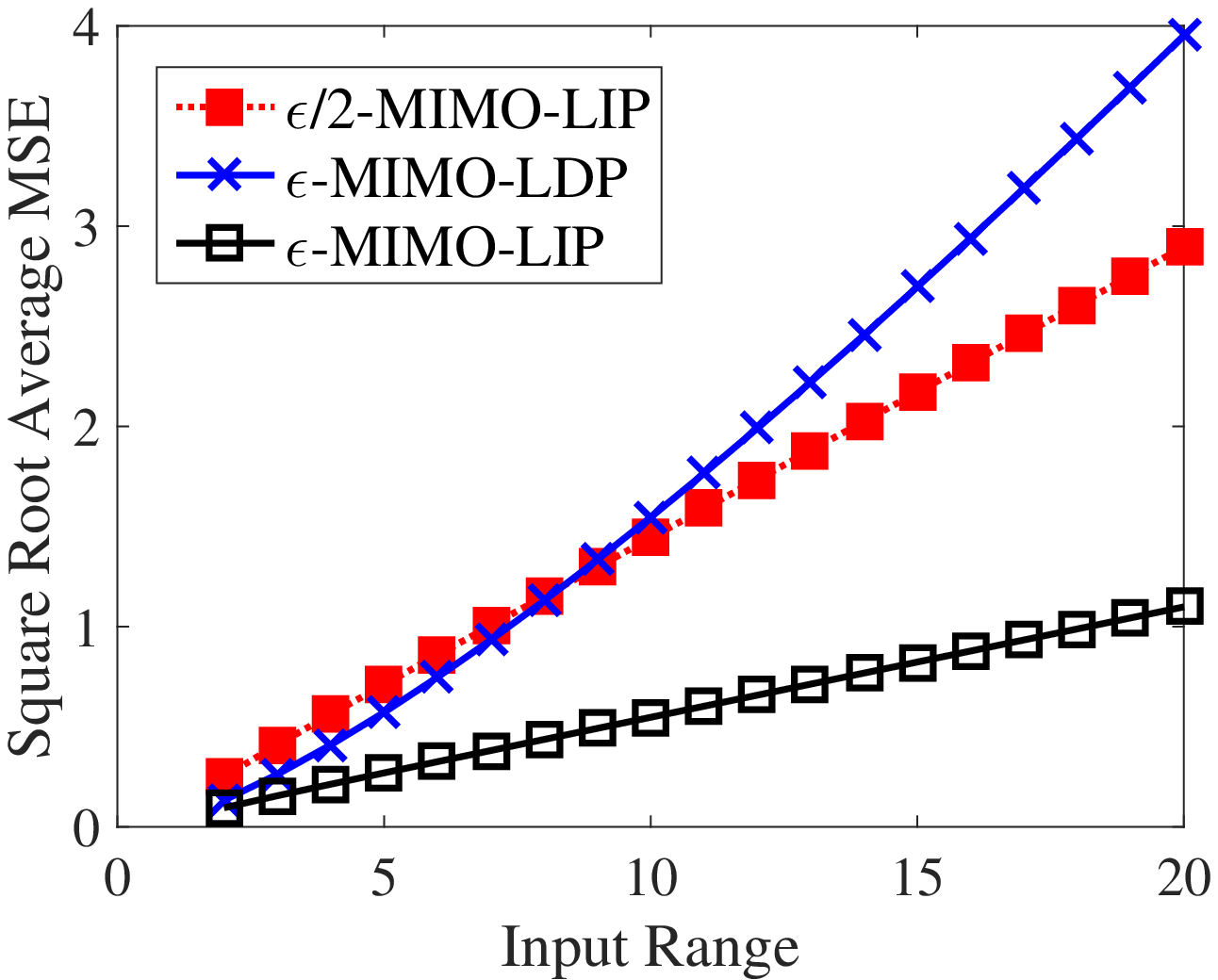} 
\label{epfixed} } 
\label{MIMO} 
\caption{Utility-privacy tradeoff comparisons under MIMO model}
\end{figure}

 \begin{figure}
 \centering
 \includegraphics[width=4.5cm]{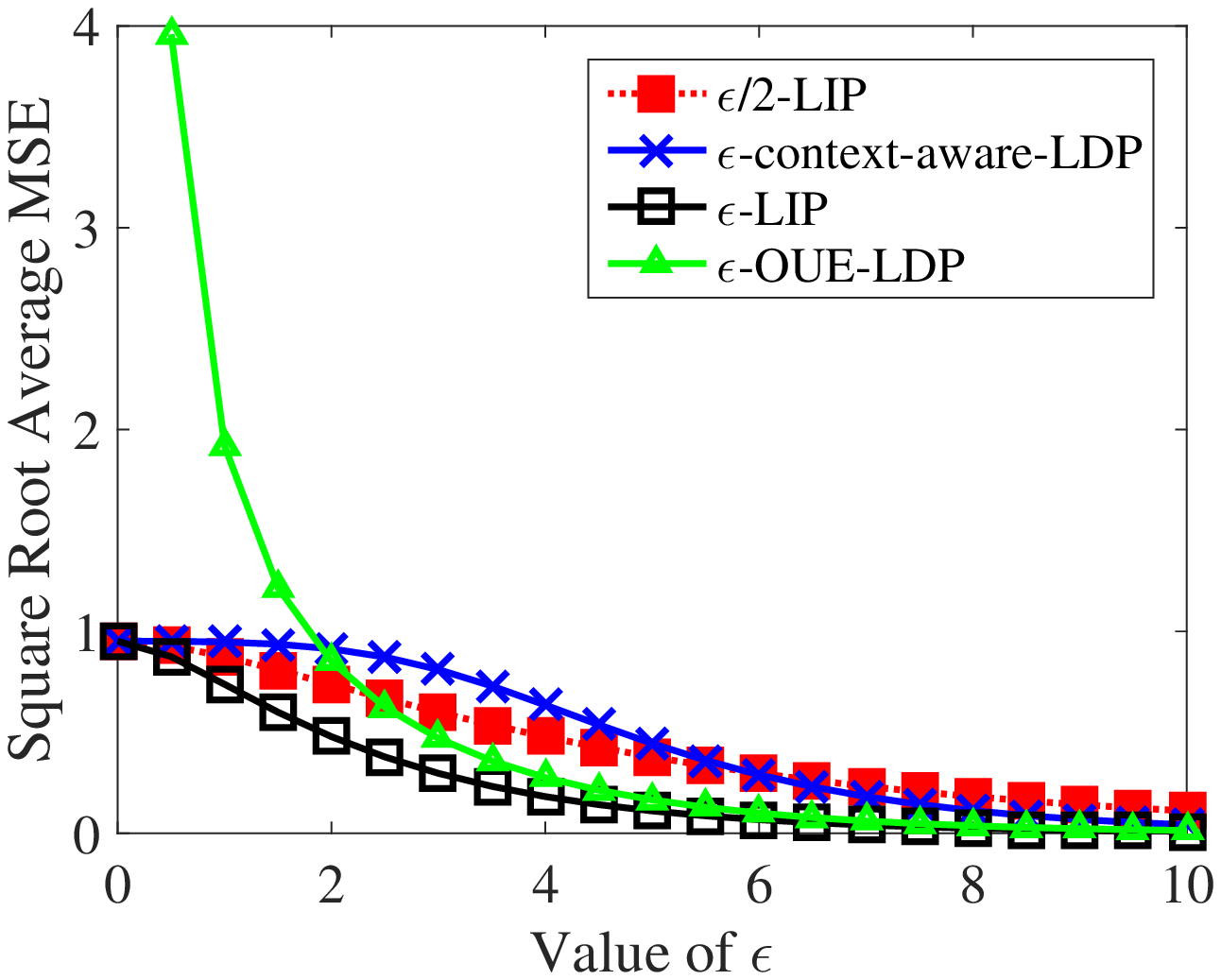}
 \label{hist_LIPLDP2}
 \caption{Utility-privacy tradeoff comparison for the application of histogram}
 \end{figure}

\begin{figure*}[htp]
\centering 
\subfigure[N=5] 
{ \includegraphics[width=5.5cm]{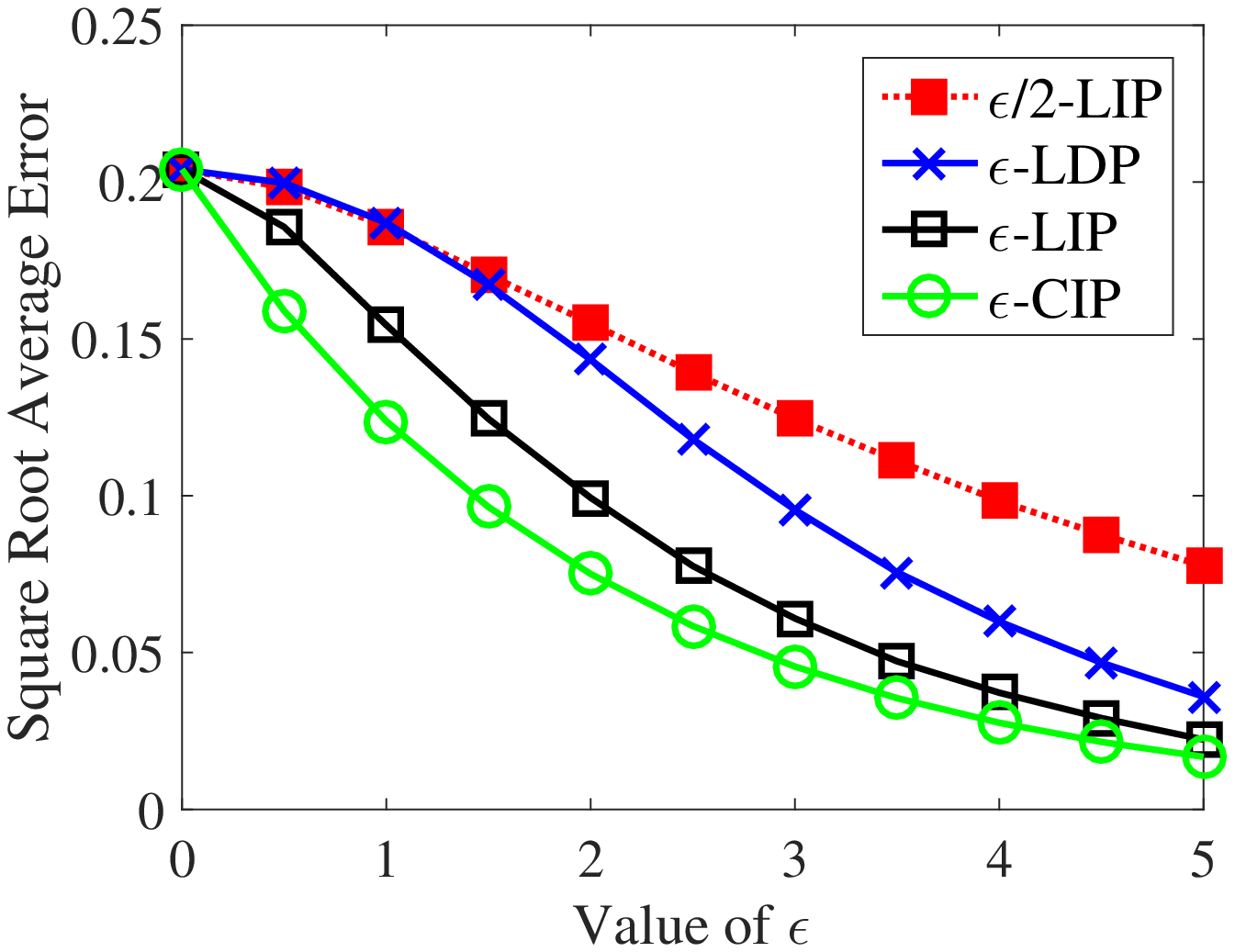} 
\label{fig_firstsub} } 
\subfigure[N=500] 
{ \includegraphics[width=5.5cm]{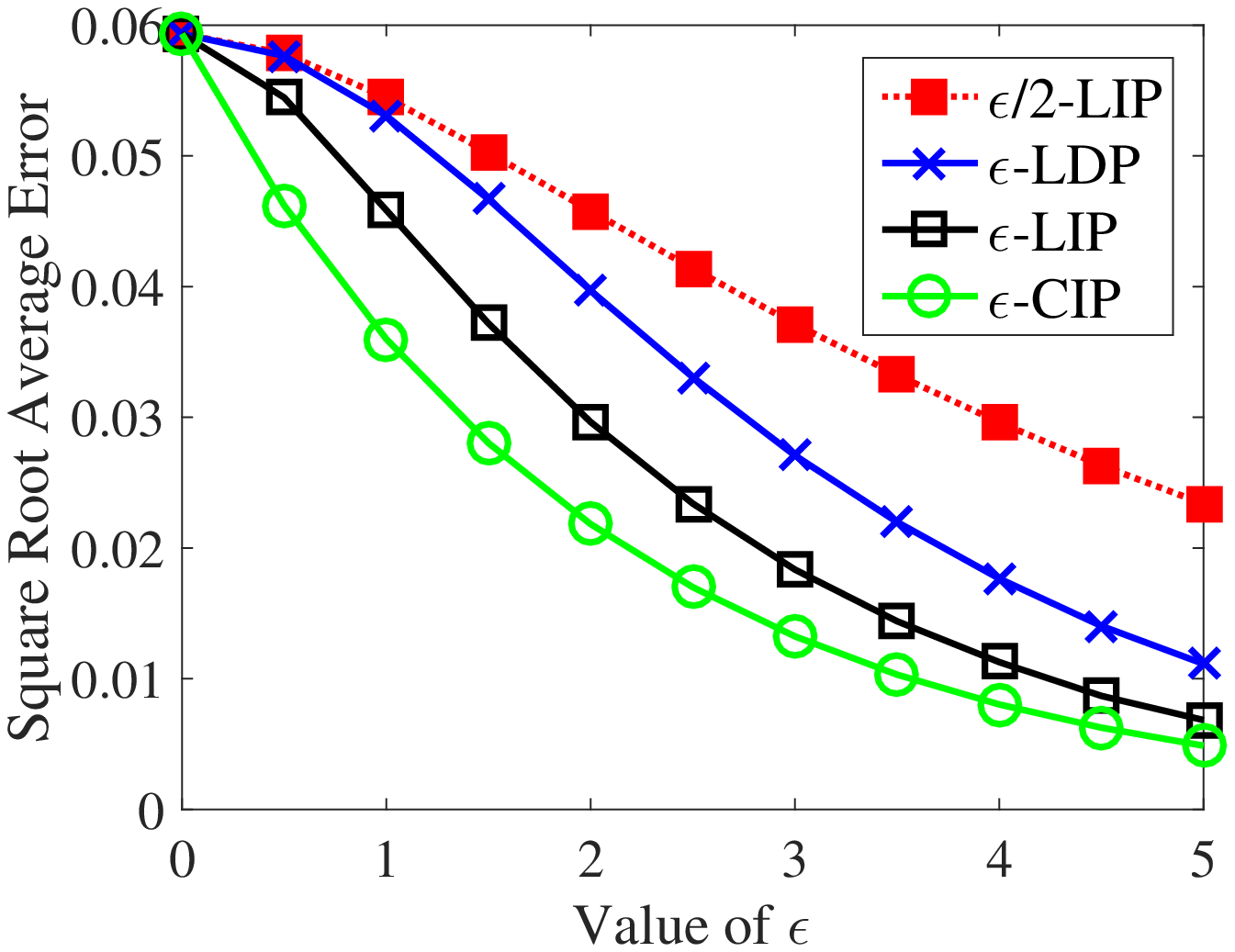} 
\label{fig_secondsub} } 
\subfigure[N=50000] 
{ \includegraphics[width=5.5cm]{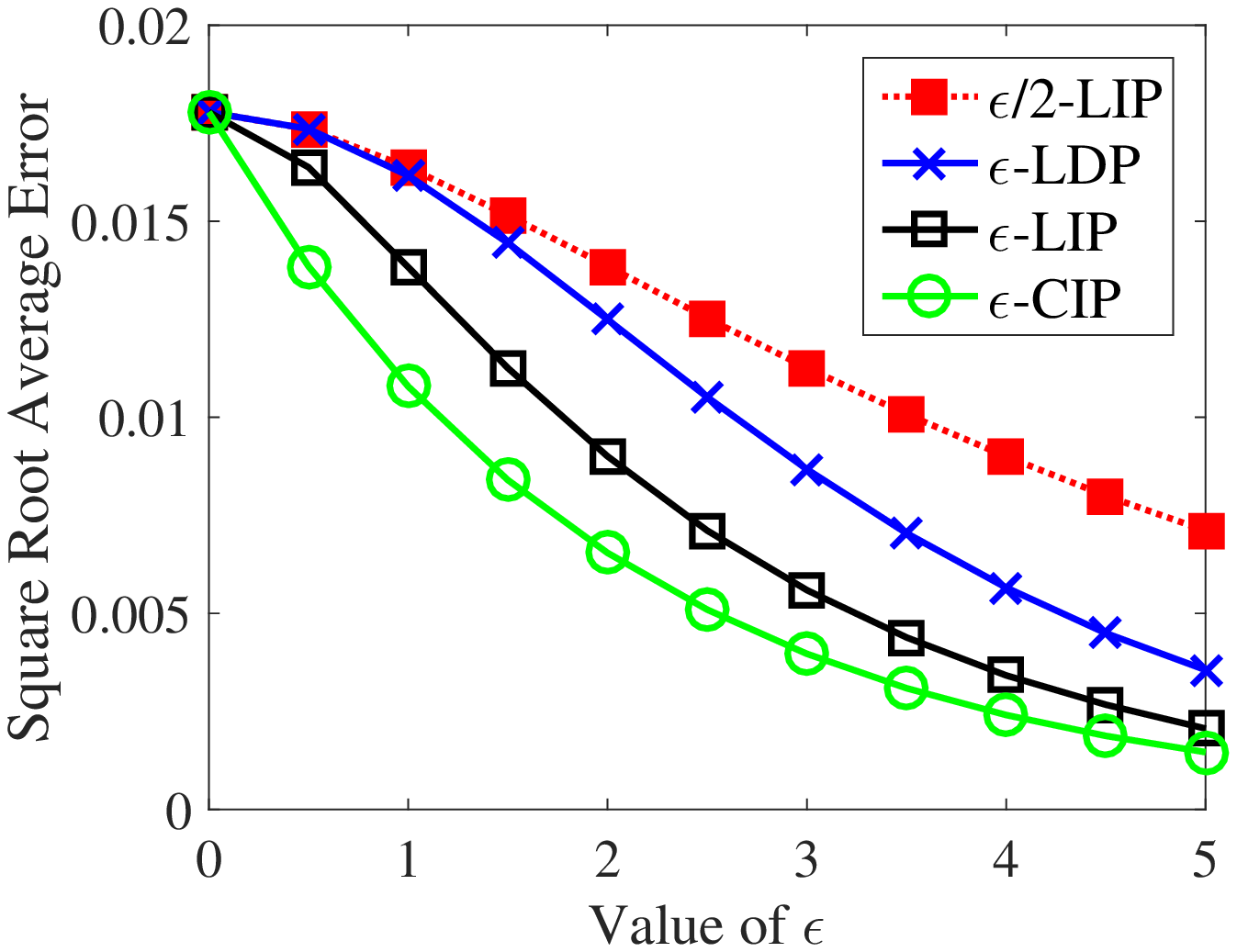} 
\label{fig_firstsub} } 
\caption{Utility-privacy tradeoff  comparison when $N$ users have different local priors, y-axis is square root average MSE: $\sqrt[]{\mathcal{E}/N}$.} 
\label{fig:Monta-Carlo} 
\end{figure*}

\subsubsection{Comparing Models  in MIMO Settings}
We now evaluate the utility-privacy tradeoffs under the optimal solutions of LIP and LDP models when the $\mathbb{D}$ has a large domain.

We first fix the users number as $N=500$ and $|\mathbb{D}|=5$. Without loss of generality, we assume that for each user, $\mathbb{D}=\{0,1,2,3,4\}$, with the prior of each value randomly generated. The utility-privacy trade offs are shown in Fig.\ref{dfixed}. We can see that the figure shows that the $\epsilon$-MIMO-LIP performs better than under the binary model, as the $\epsilon/2$-MIMO-LIP provides higher utility even than the $\epsilon$-MIMO-LDP. This means the advantage of LIP becomes more obvious with $|\mathbb{D}|$ increases. 

We next compare in detail how the data domain affects the LIP and LIP: Consider $500$ users are in the system and each of them has an input range varying from $|\mathbb{D}|=2$ to $|\mathbb{D}|=20$. To explicitly illustrate the comparison of LIP and LDP models with $|\mathbb{D}|$ increasing, we assume that each user's prior is uniformly distributed. We then fixed $\epsilon=1$, and show the utility with different input ranges under $\epsilon=1$. In Fig.\ref{epfixed}, we observe that when $|\mathbb{D}|$ is small, the $\epsilon$-MIMO-LDP model provides better utility than the $\epsilon/2$-MIMO-LIP model. However, as $|\mathbb{D}|$ increases, the $\epsilon/2$-MIMO-LIP eventually outperforms $\epsilon$-MIMO-LDP. We can also see that both the LDP and LIP models suffer from decreased utility when the input range increases, but the LIP models decreases linearly with the input range increasing while the LDP model decreases faster than that. 

\subsubsection{Histogram}
For the application of histogram, the definition of the total MSE is different from the MIMO case, as the estimator is an expected vector. We compare the performance of LIP and the LDP in two ways: (1) the LIP v.s. the context aware LDP, in this case, the utility of LDP is measured by
 the square root average MSE between the real value and the prior-related estimated vector; (2) the LIP v.s. the context-free LDP. To measure the utility of the context-free LDP, we adopt the optimal unary encoding (OUE-LDP) protocol proposed in \cite{Tianhao}.  Unary encoding is first mapping the input data into a binary vector and perturb each bit independently, which is proved to provide the highest utility among all the other LDP protocols, as the binary encoding method reduces the amount of increased error from a large domain. We fix the data domain to $|\mathbb{D}|=20$ and range $\epsilon$ from $0$ to $10$. The MSE resulted by the OUE-LDP in the application of histogram is $n\frac{4e^{\epsilon}}{(e^{\epsilon}-1)^2}$.
 The comparison is shown in Fig. 7. Observe that the tradeoff curve of the $\epsilon$-context-aware LDP is sandwiched between the $\epsilon$-LIP and the $\epsilon/2$-LIP. The LDP with optimal unary encoding results in very large MSE with small $\epsilon$, when $\epsilon$ increasing, it achieves increased utility than the context-aware LDP, but it still sandwiched between the $\epsilon$-LIP and the $\epsilon/2$-LIP. It further shows that the context-aware models are more applicable with strong privacy protection. 

\subsubsection{Monte-Carlo Simulation}
We further study the case when each user has different local priors and use Monte-Carlo Simulation to study the convergence of performance when $N$ increases. Fig. \ref{fig:Monta-Carlo} shows the comparison among three models described above as well as the model with CIP described in section \ref{sec:centralized}. We assume that each user's prior probability is sampled uniformly at random from $[0,1]$. We create three datasets with $N=5$, $500$ and $50000$, each of which contains randomly generated binary data according to their priors. For each dataset, in the localized cases, we assume that each user publishes $Y_i$ using the LIP models or the LDP models and the curator aggregates data using corresponding estimators discussed above. The error is measured by the square root of the averaged error over all users. In the centralized case, as the close form of the optimal parameters are difficult to find, we use build in optimization tools to find the parameters and derive the error. We ran each simulation 10000 times and average the errors, which are shown in Fig.\ref{fig:Monta-Carlo}. Note that when $N=5$,  curves of $\epsilon/2$-Opt-LIP and $\epsilon$-Opt-LDP cross over. This is because some users' $P^i_1$ values are far from 0.5, which makes the MSE of $\epsilon/2$-Opt-LIP smaller than that of $\epsilon$-Opt-LDP for   smaller $\epsilon$.
As $N$ becomes larger, we can see that $\epsilon$-Opt-LIP always    provides higher utility than $\epsilon$-Opt-LDP model under any $\epsilon$. We also observe that the curve of $\epsilon$-Opt-LDP   is almost sandwiched between the ones of $\epsilon/2$-Opt-LIP  and $\epsilon$-Opt-LIP. In comparison with the centralized information privacy, we can observe that the centralized information privacy always provide increased utility than the localized models.

%\subsubsection{Model with binary-input, multiple-output}
%We then explore the performance of the third special case: the model with binary input and multiple output. The simulation parameters can be found in appendix.C and the result in shown in Fig.\ref{fig:m_output}.
%\centering
%\includegraphics[width=4.5cm]{m_output.eps}\label{m_output}
%\caption{Utility-Privacy tradeoffs comparison with binary-input multiple-output model($m$ in the figure stands for the output number)}
%\centering
%\label{fig:m_output}
%\end{figure}
%Based on intuition, we know that any outputs of $Y$ other than $0$ and $1$ are aggregated as $0$ or $1$ and each of them takes certain perturbation probability. Thus any output other than $0$ and $1$ is a waste of utility. From Fig.\ref{fig:Contour}, we can see that when $m$ increases, the utility decreases. As a result, when input is binary, the optimal output range of is 2.

\subsection{Simulation with Real-world Datasets}%In this part, we run models described above with real-world data to compare with their performances.
\subsubsection{Website Popularity Statistics}
The first comparison is run using the dataset of Karosak which is a collection of anonymized click-stream data of a hungarian on-line news portal. There are around 8 million click events for 41,270 different pages. In this data set, each row stands for a click-stream for a website in different time slots.
Our goal is to estimate the frequency of popular websites (with a total click over 15,000). We treat each website as a user, thus $X_i=1$ if the total clicks of website $i$ is above 15,000; otherwise, $X_i=0$. Since no historical data is available, we regard it as the first special case in which users have a global prior.

In Fig. \ref{fig:Karosak}, we can see that $\epsilon$-Opt-LDP results in larger square root average MSE than $\epsilon$-Opt-LIP. For some smaller value of $\epsilon$, $\epsilon/2$-Opt-LIP performs better than $\epsilon$-Opt-LDP. This is because the popular websites are rare, thus the global prior is relative small. Again, this result   confirms that, the more specific the prior is, the more beneficial is LIP   than LDP. 

\begin{figure}[t]
\centering 
\subfigure[utility-privacy tradeoffs for website click aggregation with Karosak] 
{ \includegraphics[width=4cm]{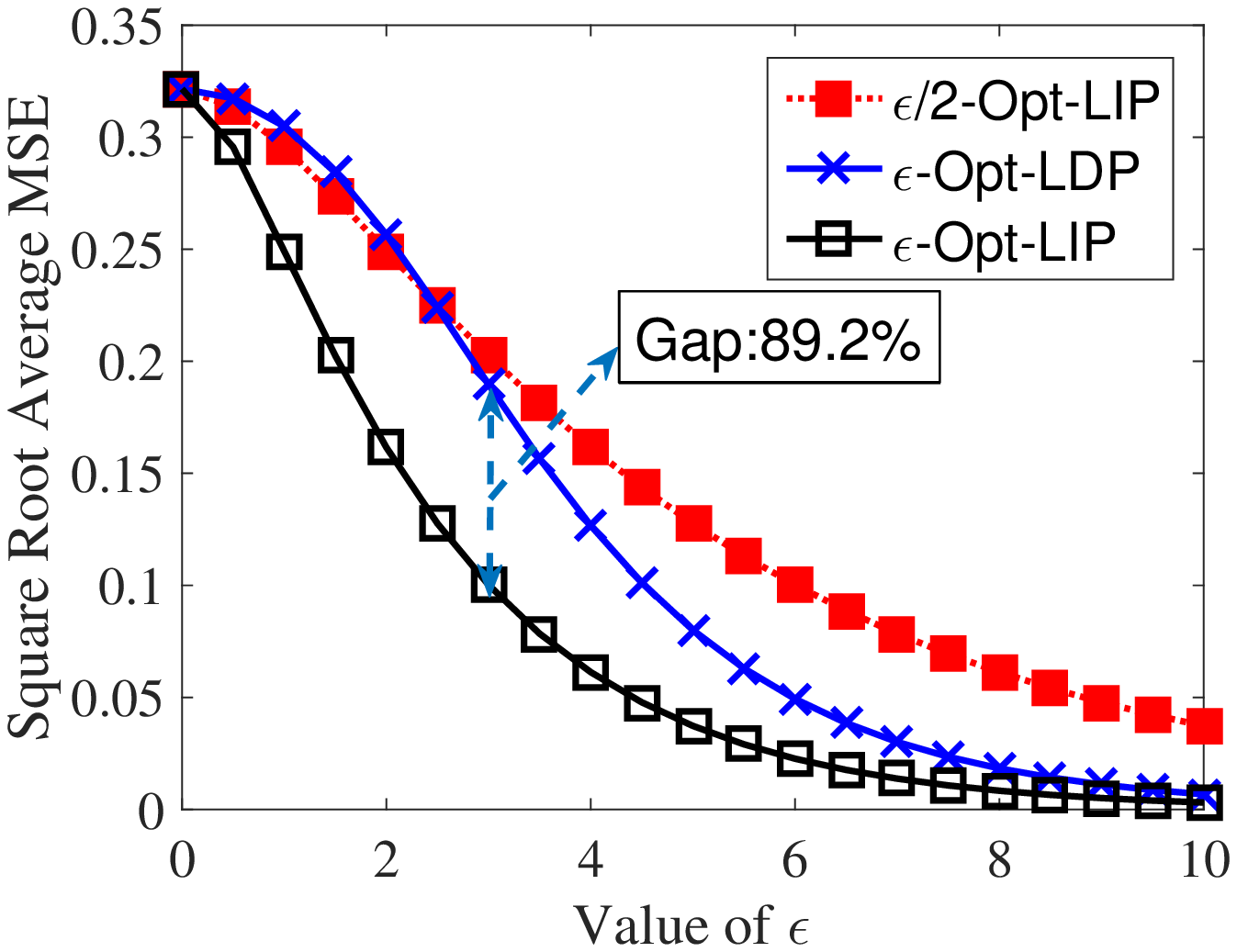} 
\label{fig:Karosak} } 
\subfigure[utility-privacy tradeoffs for location tracking with Gowalla] 
{ \includegraphics[width=4cm]{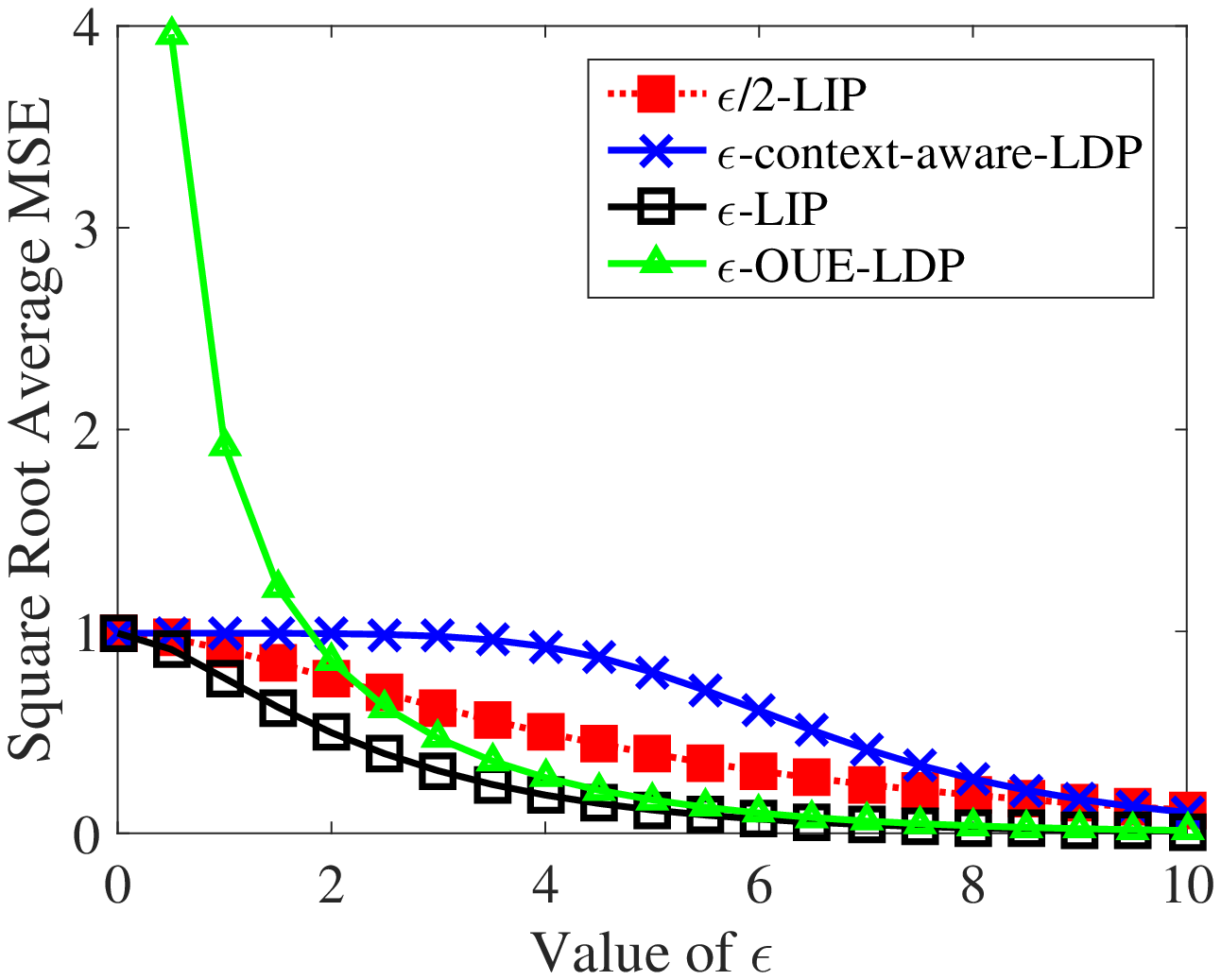} 
\label{fig:Gowalla} } \\
\caption{Utility-privacy tradeoff comparisons using real world data} 
\label{fig:single_model_compare} 
\end{figure}

\subsubsection{Location Check-In Dataset}
We then compare the performance of different models with another real-world dataset Gowalla, which is a social networking application where users share their locations by checking-in. There are 6,442,892 users in this dataset. For each user, a trace of his/her check-in locations are recorded. 
For this dataset, we wish to estimate a histogram of users' last check-in location. We first divide the area into $36\times36$ districts. We then map each user's locations into districts. Each user's past check-in locations are used for calculating a global  prior of the last check-in location for all the users. As studied in section. \ref{model_applications}, for each user, the last check-in location is perturbed according to a MIMO-LIP (LDP) channel and a random vector is used for the adversary to estimate the histogram. 
The  results are shown in Fig. \ref{fig:Gowalla}, where similar trends can be observed as in the empirical results. Note that comparing with the theoretical results from Fig. 7, the advantage of LIP is even more clear in Fig. \ref{fig:Gowalla}, that is because the theoretical analysis uses data from a domain with $|\mathbb{D}|=20$. On the other hand, in the dataset of Gowalla, the  input data is from a domain with $|\mathbb{D}|=36\times{36}$, even though many of the districts has o users checking-in, which results in zero priors for these districts. Based on the MIMO perturbation mechanism, for those districts with 0 prior, the system will also never output those districts, as a result, the data domain is equivalent to a much decreased one. Nevertheless, the effective domain size is $|\mathbb{D}|=83$, which is much larger than 20.  As we discussed, when domain size is large, the advantage of LIP is enhanced.

\section{Conclusion and Future Work}\label{sec:con}
In this paper, the notion of localized information privacy is proposed. As a context-aware privacy notion, it provides relaxed privacy guarantee than LDP by introducing prior knowledge in the privacy definition while achieving increased utility. Combined with an MMSE estimator which also leverages prior knowledge, larger gains in utility can be obtained. We studied the utility-privacy tradeoff of our proposed LIP notion and the traditional LDP notion under different both binary perturbation model and multiple-input and multiple-output perturbation model. In the binary model, we show that our $\epsilon$-LIP always outperform $\epsilon$-LDP for any given   privacy budget $\epsilon$. The advantage is more enhanced when the prior distribution is more skewed (even $\epsilon/2$-LIP with a stronger privacy guarantee than $\epsilon$-LDP is better than the latter for small $\epsilon$ values). In the MIMO model, we show that when the input data has a larger range, both the LIP mechanisms and LDP mechanism suffer decreased utility. However, as the input range increasing, the decreased amount of the LIP mechanism is less than that of the LDP mechanism. Which means the context-aware privacy notion is more applicable than the context-free privacy notion in the MIMO case.

For future work, we will extend our work to handle more general   models. For example,  we will consider an adversary that has less accurate/complete prior knowledge than users, and also understand the impact of user correlations. We will also study the optimality of the perturbation model.

\ifCLASSOPTIONcaptionsoff
  \newpage
\fi

\appendices
\section{Proof of theorem 2}
\begin{proof}
\begin{enumerate}
\item Step 1

Notice that $P^i_1$ is a constant, thus the optimization problem is equivalent to maximize:
\begin{equation}
\begin{aligned}
    L^i(q_0^{i},q_1^{i})=&\frac{(\lambda^i_0-q_1^{i})^2}{\lambda^i_0\lambda^i_1}
    =\frac{(\lambda^i_1q_1^{i}-\lambda^i_0(1-q_1^{i}))^2}{\lambda^i_0\lambda^i_1}.
\end{aligned}
\end{equation}

In order to test the monotoncity of $L^i(q_0^{i},q_1^{i})$, taking partial derivative on it.
\begin{equation}\label{eq14}
\begin{aligned}
&\frac{\partial{L^i(q_0^{i}},q_1^{i})}{\partial{q_0^{i}}}\\
    =&(1-q_1^{i})^2(\frac{\lambda^i_0}{\lambda^i_1})'+{{q_1^{i}}^2}(\frac{\lambda^i_1}{\lambda^i_0})'+[-2q_1^{i}(1-q_1^{i})]'\\
    =&(1-q_1^{i})^2(\frac{P^i_1-1}{{\lambda^i_1}^2})+{q_1^{i}}^2(\frac{1-P^i_1}{{\lambda^i_0}^2})\\
    =&(1-P^i_1)(\frac{(\lambda^i_1q_1^{i})^2-[\lambda^i_0(1-q_1^{i})]^2}{{\lambda^i_1}^2{\lambda^i_0}^2})\\
    =&(1-P^i_1)\frac{[\lambda^i_1q_1^{i}+\lambda^i_0(1-q_1^{i})][\lambda^i_1q_1^{i}-\lambda^i_0(1-q_1^{i})]}{{\lambda^i_1}^2{\lambda^i_0}^2}\\
    =&(1-P^i_1)^2\frac{[\lambda^i_1q_1^{i}+\lambda^i_0(1-q_1^{i})](q_1^{i}+q_0^{i}-1)}{{\lambda^i_1}^2{\lambda^i_0}^2}.\\
\end{aligned}
\end{equation}
In \eqref{eq14}, noticed that $\frac{\partial{L_i(q_0^{i}},q_1^{i})}{\partial{q_0^{i}}}\ge{0}$ when $q_1^{i}+q_0^{i}-1\ge{0}$; $\frac{\partial{L_i(q_0^{i}},q_1^{i})}{\partial{q_0^{i}}}<{0}$ when $q_1^{i}+q_0^{i}-1<{0}$. Which means given any value of $q_1^{i}$,  $L^i(q_0^{i},q_1^{i})$ is monotonically decreasing with $q_0^{i}$ from 0 to $1-q_1^{i}$.\\
By the same way, when $q_0^{i}$ is fixed, $L_i(q_0^{i},q_1^{i})$ is monotonically decreasing with $q_1^{i}$ from 0 to $1-q_0^{i}$.
\item Step 2

To test the symmetry of $L_i(q_0^{i},q_1^{i})$, take $q_0^{i'}=1-q_0^{i}$, $q_1^{i'}=1-q_1^{i}$.
\begin{equation*}
\begin{aligned}
L^i(q_0^{i'},q_1^{i'})=&\frac{(1-P^i_1)^2(1-q_0^i-q_1^{i})^2}{\lambda^i_1\lambda^i_0}\\
=&L^i(q_0^{i},q_1^{i}).
\end{aligned}
\end{equation*}
Thus,  $L^i(q_0^{i},q_1^{i})$ is symmetric about $(0.5, 0.5)$, which means for every point $(q_0^{i},q_1^{i})$ on the left of $q_0^{i}+q_1^{i}=1$, we can find a point on the right side of $q_0^{i}+q_1^{i}=1$ that result in a same $L^i(q_0^{i},q_1^{i})$ value.\\

In terms of the constraints, first derive $F^i_1(q_0^{i},q_1^{i})$, $F^i_2(q_0^{i},q_1^{i})$, $F^i_3(q_0^{i},q_1^{i})$, $F^i_4(q_0^{i},q_1^{i})$ as:
\begin{equation}\label{eq14}
\begin{aligned}
F^i_1(q_0^{i},q_1^{i})=&\frac{Pr(X=1)}{Pr(X=1|Y=0)}=\frac{\lambda^i_0}{q^i_1};\\
F^i_2(q_0^{i},q_1^{i})=&\frac{Pr(X=1)}{Pr(X=1|Y=1)}=\frac{\lambda^i_1}{1-q^i_1};\\
F^i_3(q_0^{i},q_1^{i})=&\frac{Pr(X=0)}{Pr(X=0|Y=0)}=\frac{\lambda^i_0}{1-q^i_0};\\
F^i_4(q_0^{i},q_1^{i})=&\frac{Pr(X=0)}{Pr(X=0|Y=1)}=\frac{\lambda^i_1}{q^i_0}.\\
\end{aligned}
\end{equation}

if taking $q_0^{i'}=1-q_0^{i}$, $q_1^{i'}=1-q_1^{i}$ into $F^i_1(q_0^{i},q_1^{i})$, $F^i_2(q_0^{i},q_1^{i})$, $F^i_3(q_0^{i},q_1^{i})$ and $F^i_4(q_0^{i},q_1^{i})$:\\ $$F^i_1(q_0^{i'},q_1^{i'})=\frac{\lambda^i_1}{1-q_1^{i}}=F^i_2(q_0^{i},q_1^{i}),$$
$$F^i_2(q_0^{i'},q_1^{i'})=\frac{\lambda^i_0}{q_1^{i}}=F^i_1(q_0^{i},q_1^{i}),$$
$$F^i_3(q_0^{i'},q_1^{i'})=\frac{\lambda^i_1}{q_0^{i}}=F^i_4(q_0^{i},q_1^{i}),$$
$$F^i_4(q_0^{i'},q_1^{i'})=\frac{\lambda^i_0}{1-q_0^{i}}=F^i_1(q_0^{i},q_1^{i}).$$
%\begin{figure}
%\centering
%\includegraphics[width=8cm]{contour.eps}\label{contour}
%\caption{Illustration of $L_i(q^i_0,q^i_1)$ and Region $\mathcal{T}_i$}
%\centering
%\label{fig:Contour}
%\end{figure}
As a result, the four constraints functions are symmetric about point $(0.5, 0.5)$. Assume that $F^i_j(q_0^{i},q_1^{i})$, $j=1,2,3,4$ form a feasible region $\mathcal{T}_i$ for $(q_0^{i},q_1^{i})$, thus for any point in this region, another point can be found on the other side of $q_0^{i}+q_1^{i}=1$ also in $\mathcal{T}_i$.
Now $(q_0^{i*},q_1^{i*})\in{\mathcal{T}_i=\mathcal{T}_i\cap\{{q_0^{i}+q_1^{i}\le{1}}\}}$.

Fig. \ref{contour} illustrates the function of $L_i(q_0^{i},q_1^{i})$ and feasible region $\mathcal{R}^i$ on the plane of $q_0^{i},q_1^{i}$, curves in the figure is the contour line of $L_i(q_0^{i},q_1^{i})$, while the shadow region is the feasible region of $\mathcal{T}_i$ when $\epsilon$ is fixed.
\item Step 3

Since 
\begin{equation*}
\begin{aligned}
F^i_1-F^i_2=&P^i_1+\frac{(1-P^i_1)(1-q^i_0)}{q^i_{1}}-(P^i_1+\frac{(1-P^i_1)q^i_0}{1-q^i_{1}})\\
=&(1-P^i_1)\frac{1-q^i_0-q^i_1}{q^i_1(1-q^i_1)}\\
\end{aligned}
\end{equation*}
and
\begin{equation*}
\begin{aligned}
F^i_4-F^i_3=&1-P^i_1+\frac{P^i_1(1-q^i_1)}{q^i_{0}}-(1-P^i_1+\frac{P^i_1q^i_1}{1-q^i_0})\\
=&(1-P^i_1)\frac{1-q^i_0-q^i_1}{q^i_1(1-q^i_1)},\\
\end{aligned}
\end{equation*}
for any point $(q_0^{i},q_1^{i})\in{\mathcal{T}_i}$, there is $F^i_1\ge{F^i_2}$ and $F^i_4\ge{F^i_3}$\\
\begin{equation}\label{new_region}
\begin{aligned}
\mathcal{T}_i=\{F^i_1\le{e^{\epsilon}}\}&\cap\{e^{-\epsilon}\le{F^i_2}\}\cap\{e^{-\epsilon}\le{F^i_3}\}\\
&\cap\{F^i_4\le{e^{\epsilon}}\}\cap\{{q_0^{i}+q_1^{i}\le{1}}\}
\end{aligned}
\end{equation}
From \eqref{new_region}, assume that $q_0^{i}$ is fixed to be $q_0$, then 
$$q_1^i\ge
\begin{cases}
\frac{(1-P^i_1)(1-q_0)}{e^{\epsilon}-P^i_1}& if \quad  {q_0\ge{P^i_1/e^{\epsilon}}}\\
1-\frac{(e^{\epsilon}+P^i_1-1)(q_0)}{P^i_1}& if \quad  {q_0<{P^i_1/e^{\epsilon}}}.
\end{cases}$$

Since $L^i(q_0^{i},q_1^{i})$ is monotonically decreasing with $q_0^{i}$ and $q_1^{i}$ in $\mathcal{T}_i$, For any $q_0^{i}$ in $\mathcal{T}_i$, its according optimal $q_1^{i*}$ is the smallest value s.t. $(q_0^{i},q_1^{i*})$ is in $\mathcal{T}_i$. It is the same with $q_0^{i*}$. On the other hand, $F^i_1$, $F^i_2$, $F^i_3$, $F^i_4$ are all liner, their bounds value can be achieved when they are equal to their constraints values $e^{\epsilon}$ or $e^{-\epsilon}$.
As a result, 
$$(q^{i*}_0,q^{i*}_1)\in
\begin{cases}
\{F^i_1={e^{\epsilon}}\}& if \quad  {q_0\ge{P^i_1/e^{\epsilon}}}\\
\{F^i_4={e^{\epsilon}}\}& if \quad  {q_0<{P^i_1/e^{\epsilon}}}.
\end{cases}$$
\item Step 4

To find the minimal value of $L^i(q_0^{i},q_1^{i})$, where $(q_0^{i},q_1^{i})\in \mathcal{T}_i$, the last thing we need to do is to test its monotocity of $L^i(q_0^{i},q_1^{i})$ given $F^i_1={e^{\epsilon}}$ or $F^i_4=e^{\epsilon}$.
\begin{equation*}
\begin{aligned}
L^i_{F^i_1={e^{\epsilon}}}(q_0^{i},q_1^{i})=&\frac{(1-P^i_1)^2(q_0^i+q_1^{i}-1)^2}{\lambda^i_0\lambda^i_1}\\
=&\frac{[(1-P^i_1)(e^{\epsilon}-1)q_1^i]^2}{\lambda^i_0\lambda^i_1},
\end{aligned}
\end{equation*}
To find the minimal value, taking derivative over $L^i_{F^i_1={e^{\epsilon}}}$
\begin{equation*}
\begin{aligned}
&\frac{\partial{L^i_{F^i_1={e^{\epsilon}}}(q_0^{i},q_1^{i})}}{\partial{q_1^i}}\\
=&(1-P^i_1)^2(e^{\epsilon}-1)^2(\frac{{q_1^i}^2}{\lambda^i_0\lambda^i_1})'\\
=&(1-P^i_1)^2(e^{\epsilon}-1)^2{q_1^i}\frac{\lambda^i_1(\lambda^i_0-P^i_1q^i_1)+\lambda^i_0\lambda^i_1+P^i_1q^i_1\lambda^i_0}{(\lambda^i_0\lambda^i_1)^2}
\end{aligned}
\end{equation*}
Which is obviously greater than 0, similarly, when $F^i_4=e^{\epsilon}$ is given:
\begin{equation*}
\begin{aligned}
&\frac{\partial{L^i_{F^i_4={e^{\epsilon}}}(q_0^{i},q_1^{i})}}{\partial{q_0^i}}\\
=&{P^i_1}^2(e^{\epsilon}-1)^2(\frac{{q_0^i}^2}{\lambda^i_0\lambda^i_1})'\\
=&{P^i_1}^2(e^{\epsilon}-1)^2{q_0^i}\frac{\lambda^i_0(\lambda^i_1-(1-P^i_1)q^i_0)+\lambda^i_1(1-P^i_1)q^i_0}{(\lambda^i_0\lambda^i_1)^2}
\end{aligned}
\end{equation*}
Which is also greater then 0. Noticed that, when $q_0^i\le{\frac{P^i_1}{e^{\epsilon}}}$, optimal point is $(q^{i*}_0,q^{i*}_1)=(\frac{P^i_1}{e^{\epsilon}},\frac{1-P^i_1}{e^{\epsilon}})$, when $q_0^i>{\frac{P^i_1}{e^{\epsilon}}}$, the optimal point is also $(q^{i*}_0,q^{i*}_1)=(\frac{P^i_1}{e^{\epsilon}},\frac{1-P^i_1}{e^{\epsilon}})$. Thus the global optimal solution is $(\frac{P^i_1}{e^{\epsilon}},\frac{1-P^i_1}{e^{\epsilon}})$.
\end{enumerate}
\end{proof}

\section{Proof of Theorem 3}
\begin{proof}
Notice that $Var[X_i]$ is a non-negative constant, thus minizing MSE is equivalent to maximize $Var[\hat{X_i}]$.
\begin{lem}
Regardless of the privacy constraints, one set of solutions of $\mathbf{q^i}$ that result in the minimal value of $Var[\hat{X_i}]$ is $q^i_{nk}=\lambda^i_k$, $\forall{n,k\in{1,2,3...d}}$; one set of solutions of $\mathbf{q^i}$ that result in the maximal value of $Var[\hat{X_i}]$ is all $q^i_{kj}$s are from the set of $\{0,1\}$, and for any $n\neq{m}$, $q^i_{nj}$ and $q^i_{mj}$ can not be 1 simultaneously. 
\end{lem}
\begin{proof}
Obviously, 
\textbf{minimized solution:}\\
Consider a set of parameters: $\mathbf{q_{min}^i}$, when $q^i_{nk}=\lambda^i_k$, $\forall{n,k\in{1,2,3...d}}$, $Var[\hat{X_i}]=0$, as $Var[\hat{X_i}]\ge0$, thus $q^i_{nk}=\lambda^i_k$ results in a minimal value of  $Var[\hat{X_i}]$.\\
On the other hand, \textbf{maximized solution:} Consider a set of parameters: $\mathbf{q_{max}^i}$, assume that for all $n\in{1,2...d}$, there's a $k\in{1,2,...,d}$, such that $q^i_{nk}=1$ and $q^i_{nl}=0$ for all $l\neq{k}$ and for different $n$, $k$ is different, thus $\lambda^i_k=P^i_n$
\begin{equation}
\begin{aligned}
    &\sum^d_{m=1}\sum^d_{n=1}\sum^d_{k=1}a_ma_nP^i_mP^i_nq^i_{mk}(\frac{q^i_{nk}}{\lambda^i_k}-1)\\
    =&\sum^d_{m=1}\sum^d_{n=1}a_na_mP^i_nP^i_mq^i_{mk}(\frac{q^i_{nk}}{P^i_n}-1)\\
    =&\sum^d_{n=1}a^2_nP^i_n(1-P^i_n)-\sum^d_{n=1}\sum^d_{m\neq{n}}a_na_mP^i_nP^i_m\\
    =&Var[X_i]
    \end{aligned}
\end{equation}
Notice that $MSE_i\ge{0}$, $Var[X_i]\ge{Var[\hat{X_i}]}$. Thus $\mathbf{q_min}^i$ is a maximum when for all $n\in{1,2,...,d} $, $q^i_{nk}=1$ and $q^i_{nl}=0$ for all $l\neq{k}$.\\
When $\epsilon=0$, the mechanism satisfies strongest privacy guarantee and $\frac{\lambda^i_k}{q^i_{jk}}=1$. We know that when $q^i_{nk}=\lambda^i_k$, the mechanism achieves minimized utility. 
\end{proof}

We now develop monotocity of the region between minimum and maximum with the following lemma:
\begin{lem}
The sets of solutions in the form in lemma.1 that result in the maximum and minimum of $Var[\hat{X_i}]$ are unique. On the other hand, $Var[\hat{X_i}]$  is monotonically increasing when $q^i_{lk}>\lambda^i_k$;  $Var[\hat{X_i}]$ is monotonically decreasing when $q^i_{lk}<\lambda^i_k$
\end{lem}
\begin{proof}

Taking derivative over $q^i_{lk}$, where $l,k\in{1,2...,d}$:
\begin{equation}\label{der}
\begin{aligned}
   &\frac{\partial{Var[\hat{X_i}]}}{\partial{q^i_{lk}}}\\
   =&\frac{1}{(\lambda^i_k)^2}[a_l\lambda^i_k(2\sum^d_{m=1}(a_mq^i_{mk}-a_j\lambda^i_k))-P^i_l(\sum^d_{m=1}a^2_m(q^i_{mk})^2\\
   &\quad\qquad-2\sum^d_{m=1}\sum^d_{n=1}a_ma_nq^i_{mk}q^i_{nk}]\\
   =&\frac{1}{(\lambda^i_k)^2}[a_l\lambda^i_k(2\sum^d_{m=1}(a_mq^i_{mk}-a_j\lambda^i_k))-P^i_l(\sum^d_{m=1}a_mq^i_{mk})^2]\\
   =&\frac{a_{l}q^i_{lk}(\sum_{m\neq{l}}^da_mq^i_{mk})(1-P^i_k)(q^i_{lk}-\lambda^i_k)}{\lambda^i_k}.
\end{aligned}
\end{equation}
From \eqref{der}, we can see that the station point of $q^i_{lk}$ is $\lambda^i_k$, which we know is the minimal value and $Var[\hat{X_i}]$ is monotonically increasing when $q^i_{lk}>\lambda^i_k$;  $Var[\hat{X_i}]$ is monotonically decreasing when $q^i_{lk}<\lambda^i_k$. As a result, without considering the privacy constraints, the optimal solutions of each $q^i_{mn}$ is either $0$ or $1$. We first show that the maximum value of $Var[\hat{X_i}]$ can only be achieved by the solutions discussed above.

We take one set of optimal solution as an examples, other solutions follow the same idea, the set of optimal solution is: $q^i_{kk}=1$ for any $k\in{1,2,...,d}$, and $q^i_{kj}=0$ for any $j\in{1,2,...,d}$. Know we assume that there is an subset of index $1$ to $n$, s.t: $q^i_{ll_k}\neq{1}\neq{0}$, for any $k\in\{1,2,...,n\}$. By the monotocity, we know that, regardless the privacy constraint and the total probability constraints, $Var[\hat{X_i}]$ with $q^i_{lk}\neq{1}\neq{0}$ and is less than  $Var[\hat{X'_i}]$ with $q^i_{lk}=1$, for any $k\in\{1,2,...,n\}$. Compare with the two solutions, we have:
\begin{equation}
\begin{aligned}
    &Var[X_i]-Var[\hat{X'_i}]\\=&\sum^n_{k=1}a^2_lP^i_l(\frac{P^i_l}{P^i_l+P^i_k})+\sum^n_{k=1}a^2_kP^i_k(\frac{P^i_k}{P^i_l+P^i_k})\\
   +& \sum^d_{m\notin\{1,2,...,n\}}a_lP^i_la_mP^i_m-2\sum^n_{k=1}a_la_k\frac{P^i_lP^i_k}{P^i_l+P^i_k}\\
    =&\sum^n_{k=1}\frac{(a_lP^i_l-a_kP^i_k)^2}{P^i_l+P^i_k}+\sum^d_{m\notin\{1,2,...,n\}}a_lP^i_la_mP^i_m\\
    &>{0}
\end{aligned}
\end{equation}

Thus the optimal solution is: only one of the $q^i_{kk}$ can be one, other $q^i_{kj}$s are all zeros.

\end{proof}

Now, we know that the optimal solution for the miniming MSE problem with $\epsilon$ growing is from the minimal value to its maximal value. As a result of the monotocity property, the optimal solution (with privacy constraints) lies on the boundaries of the constraints:     $e^{-\epsilon}=\frac{\lambda^i_k}{q^i_{jk}}$, or $\frac{\lambda^i_k}{q^i_{jk}}={e^{\epsilon}}$ as well as the constraints: $0\le{q^i_{jk}}$; $\sum^d_{n=1}q^i_{jn}=1;$, $\forall{j,k\in{1,2,...,d}}$.

Based on the solutions which result in the maximum and minimum of $Var[\hat{X_i}]$. We know that one of the probabilities of $q^i_{m1}, q^i_{m2},..., q^i_{md}, \forall{m}$ approaches 1 and others approaches 0. As we can randomly permute the sequence of $Y_i$, there are $d!$ feasible solutions. We now consider the case in which $q^i_{kk}$s approach 1 for all $k\in{1,2,...,d}$, and other $q^i_{kj}$s are approaching  0, where $j\neq{k}$. For the $q^i_{kk}$s which approach 1, the upper bounds is restricted, and for $q^i_{kj}$s which approach 0, the lower bounds are mounted. Considering the privacy constraints, we know the upper of $q^i_{kk}$ is $\lambda^i_k/e^{-\epsilon}$ and the lower bound of $q^i_{kj}$ is $\lambda^i_k/e^{\epsilon}$. As  $q^i_{kk}+\sum_{j=1,j\neq{k}}^dq^i_{kj}=1$, for all $j$s $q^i_{kj}$s are approaching boundaries simultaneously, as a result, they may not reach the boundaries at the same time.

We now discuss whether lower bounds or upper bounds are reached first.

\begin{lem}
The parameters that decrease when $\epsilon$ grows reach the boundary first.
\end{lem}
\begin{proof}

When lower bounds are reached, $q^i_{j,k}=\frac{\lambda_k}{e^{\epsilon}}$ for all $j,k\in{1,2,3,...,d}, j\neq{k}$. Thus $q^i_{kk}=1-(1-P^i_k)/e^{\epsilon}$.
\begin{equation}
\begin{aligned}
    \lambda^i_k=&\sum_{j\neq{k}}^d\frac{P^i_k}{e^{\epsilon}}+P^i_k(1-\frac{(1-P^i_k)}{e^{\epsilon}})\\
    =&(1-P^i_k)\frac{P^i_k}{e^{\epsilon}}+P^i_k(1-\frac{(1-P^i_k)}{e^{\epsilon}})\\
    =&P^i_k
    \end{aligned}
\end{equation}
We can check whether $q^i_{kk}$s are in the feasible region:
\begin{equation}
\begin{aligned}
    \frac{\lambda^i_k}{q^i_{kk}}-e^{-\epsilon}=&\frac{e^{\epsilon}P^i_k}{e^{\epsilon}+P^i_k-1}-e^{-\epsilon}\\
    =&\frac{1-e^{-\epsilon}+P^i_k(e^{\epsilon}-e^{-\epsilon})}{e^{\epsilon}+P^i_k-1}\ge{0}
    \end{aligned}
\end{equation}
\begin{equation}
\begin{aligned}
    e^{\epsilon}-\frac{\lambda^i_k}{q^i_{kk}}=&e^{\epsilon}-\frac{e^{\epsilon}P^i_k}{e^{\epsilon}+P^i_k-1}\\
    =&\frac{e^{\epsilon}(e^{\epsilon}-1)}{e^{\epsilon}+P^i_k-1}\ge{0}
    \end{aligned}
\end{equation}
So, we know that when $q^i_{kj}$s reach the lower bound, $q^i_{kk}$ is still in the feasible region, it's easy to check that when $q^i_{kk}$ reaches the upper bound, $q^i_{kj}$s do not satisfies the privacy constraints. 
\end{proof}

As a result, one optimal solution is: $q^i_{kj}=P^i_{j}/e^{\epsilon}$ for all $k,j\in{1,2,3...d}$ and $j\neq{k}$; $q^i_{kk}=1-(1-P^i_{k})/e^{\epsilon}$ for all  $k\in{1,2,3...d}$. We can randomly permute the sequence of $Y_i$, thus there are $d!$ optimal solutions.
\end{proof}

\section{Proof of Theorem.4}
\begin{proof}
 We know the optimal solution of the parameters of any input $a_k$ are in the form of: $q^i_{kk}$ is approaching to 1 while other $q^i_{kj}$s are approaching to 0 so that each input value can be inferred by a particular output. For example, given $Y^i=a_k$, we can probably infer that $X^i$ is also $a_k$ and the confidence increases with $\epsilon$. 
 \begin{itemize}
     \item  Assume that $f<d$, intuitively, this will achieve decreased utility because at least one of input values with index $k\in\{1,2,...,d\}$ can not find an parameter that approaches to 1, as a result, inputs with indexed $\{f+1,f+2,...,d\}$ can not be inferred by any outputs. 
     
     Mathematically, as the $d$ is fixed, $Var(X)$ is also fixed. denote  $Var(\hat{X_i})$ as the variance of the estimator with $d=f$ and $Var(\hat{X'_i})$ as the variance of the estimator with $d>f$.
     recall that 
     \begin{equation}
        \hat{X_i}=\sum^d_{j=1}\sum^d_{k=1}a_jPr(X_i=a_j|Y_i=a_k)\mathbbm{1}^i_{k},  
     \end{equation}
     derive the $\hat{X'_i}$:
     \begin{equation}
        \hat{X'_i}=\sum^d_{j=1}\sum^f_{k=1}a_jPr(X_i=a_j|Y_i=a_k)\mathbbm{1}^i_{k},  
     \end{equation}
     First assume that for each $j\in\{1,2,...,d\}$, $k\in\{1,2,...,f\}$, the parameters of  $\hat{X_i}$ and $\hat{X'_i}$ are identical. We know that for each $j\in\{1,2,...,d\}$, $k\in\{1,2,...,f\}$, $a_jPr(X_i=a_j|Y_i=a_k)\ge{0}$, thus $Var(\hat{X'}_i)$ is monotonically increasing with $f$. 
     
     Notice that the parameters of $\hat{X_i}$ and $\hat{X'_i}$ can not be identical as for at least one $j$, $q^i_{kj}$ will increase for $k\in\{f+1,f+2,...,d\}$, $j\in\{1,2,...,f\}$. However, the increase of these values will make each $Pr(X_i=a_k|Y_i=a_j)$ smaller, thus  $Pr(X_i=a_k|Y_i=a_j)>Pr(X'_i=a_k|Y'_i=a_j)$.
     
     As a result: $Var(X_i)>Var(X'_i)$.
     \item Assume that $d<f$, this case can be viewed as a special case of the general model with $P^i_{d+1}=P^i_{d+2}=...=P^i_{f}$. Thus the optimal solutions is straightforward: $q^i_{kk}=1-(1-P^i_k)/e^{\epsilon}$, $q^i_{kj}=P^i_j/e^{\epsilon}$ for $k,j\in\{1,2,...,d\}$; $q^i_{kj}=0$, for $k\in\{1,2,...,d\}$, $j\in\{d+1,d+2,...,f\}$. As a result, the optimal solution is equivalent to the case of the general model with $d=f$.
 \end{itemize}
 In summary, the optimal range of output is $f=d$.

\end{proof}
\section{Proof of theorem 5}\label{sec:proof_histomgram}
\begin{proof}
The parameter-related term in \eqref{Error} can be further expressed as: 

\begin{equation}\label{MSE_histogram}
\begin{aligned}
=&\sum^N_{i=1}\sum^d_{k=1}\{Var(\sum^d_{j=1}Pr(X_i=a_k|Y_i=a_j)\mathbbm{1}_{\{Y_i=a_j\}})\}\\
=&\sum^N_{i=1}\sum^d_{k=1}\{(\sum^d_{j=1}Pr(X_i=a_k|Y_i=a_j)^2Var(\mathbbm{1}_{\{Y_i=a_j\}}))\\
+&\sum^d_{j=1}\sum^d_{l\neq{j}}Pr(X_i=a_k|Y_i=a_j)Pr(X_i=a_k|Y_i=a_l)\\
&\cdot{Cov(\mathbbm{1}_{\{Y_i=a_j\}}\cdot{\mathbbm{1}_{\{Y_i=a_l\}}})\}}\\
=&\sum^N_{i=1}\sum^d_{k=1}\{[\sum^d_{j=1}\frac{(q^i_{kj}P^i_{k})^2}{\lambda^i_j}(1-\lambda^i_j)-\sum^d_{j=1}\sum^d_{l\neq{j}}q^i_{kj}q^i_{kl}(P^i_{k})^2]\}\\
=&\sum^N_{i=1}\sum^d_{k=1}\{(P^i_{k})^2\sum^d_{j=1}q^i_{kj}(\frac{q^i_{kj}}{\lambda^i_j}-1)\}\\
\end{aligned}
\end{equation}
Comparing with the parameter-related term in the MSE formulation of the mimo model:
\begin{equation}\label{parameter_mimo}
  \sum^N_{i=1}\sum^d_{m=1}\sum^d_{n=1}\sum^d_{j=1}a_ma_nP^i_mP^i_nq^i_{mj}(\frac{q^i_{nj}}{\lambda^i_j}-1).
\end{equation}
when each $m=n=k$, and $a_j=1$ for $j\in\{1,2,...,d\}$, \eqref{parameter_mimo} becomes:
\begin{equation}\label{eq_mn}
    \sum^N_{i=1}\sum^d_{k=1}\sum^d_{j=1}(P^i_k)^2q^i_{kj}(\frac{q^i_{kj}}{\lambda^i_k}-1)=\sum^N_{i=1}\sum^d_{k=1}(P^i_k)^2\sum^d_{j=1}q^i_{kj}(\frac{q^i_{kj}}{\lambda^i_k}-1).
\end{equation}
Notice that \eqref{eq_mn} is the same with the parameter-related term in  and we can consider $a_j=1$, $\forall{j\in\{1,2,...,d\}}$, $m=n=k$ as a special case of \eqref{MSE_histogram}; On the other hand, these two minimization problem has the same privacy constraints. Thus their optimal solutions are identical.

\end{proof}


\begin{thebibliography}{10}

\bibitem{Jian1805:Context}
B.~Jiang, M.~Li, and R.~Tandon, ``{Context-Aware} data aggregation with
  localized information privacy,'' in {\em 2018 IEEE Conference on
  Communications and Network Security (CNS) (IEEE CNS 2018)}, (Beijing, P.R.
  China), May 2018.

\bibitem{SS98}
P.~Samarati and L.~Sweeney, ``Protecting privacy when disclosing information:
  k-anonymity and its enforcement through generalization and suppression,''
  tech. rep., 1998.

\bibitem{Dwork2008}
C.~Dwork, ``Differential privacy: A survey of results,'' in {\em Theory and
  Applications of Models of Computation: 5th International Conference, TAMC}
  (M.~Agrawal, D.~Du, and Z.~Duan, eds.), pp.~1--19, 2008.

\bibitem{Dwork20061}
C.~Dwork, ``Differential privacy,'' in {\em Automata, Languages and
  Programming: 33rd International Colloquium, ICALP 2006, Part II}
  (M.~Bugliesi, B.~Preneel, V.~Sassone, and I.~Wegener, eds.), pp.~1--12, 2006.

\bibitem{Dwork2006}
C.~Dwork, F.~McSherry, and K.~Nissim, ``Calibrating noise to sensitivity in
  private data analysis,'' in {\em Theory of Cryptography: Third Theory of
  Cryptography Conference}, pp.~265--284, 2006.

\bibitem{Yahoo}
A.~FITZPATRICK, ``What to do after the massive yahoo hack,'' 2016.

\bibitem{aggregation}
R.~Chen, H.~Li, A.~K. Qin, S.~P. Kasiviswanathan, and H.~Jin, ``Private spatial
  data aggregation in the local setting,'' in {\em 2016 IEEE 32nd ICDE},
  pp.~289--300, May 2016.

\bibitem{randomresponse}
S.~L. Warner, ``Randomized response: A survey technique for eliminating evasive
  answer bias,'' {\em Journal of the American Statistical Association},
  vol.~60, no.~309, pp.~63--69, 1965.

\bibitem{Freudiger:2011:EPR:2186383.2186387}
J.~Freudiger, R.~Shokri, and J.-P. Hubaux, ``Evaluating the privacy risk of
  location-based services,'' in {\em Proceedings of the 15th International
  Conference on Financial Cryptography and Data Security}, FC'11, (Berlin,
  Heidelberg), pp.~31--46, Springer-Verlag, 2012.

\bibitem{Extreme_ldp}
P.~Kairouz, S.~Oh, and P.~Viswanath, ``Extremal mechanisms for local
  differential privacy,'' in {\em Advances in Neural Information Processing
  Systems 27}, pp.~2879--2887, Curran Associates, Inc., 2014.

\bibitem{ldp_lalitha}
A.~D. Sarwate and L.~Sankar, ``A rate-disortion perspective on local
  differential privacy,'' in {\em 2014 52nd Annual Allerton Conference on
  Communication, Control, and Computing}, pp.~903--908, Sept 2014.

\bibitem{rr_ldp}
S.~Xiong, A.~D. Sarwate, and N.~B. Mandayam, ``Randomized requantization with
  local differential privacy,'' in {\em 2016 IEEE International Conference on
  Acoustics, Speech and Signal Processing (ICASSP)}, pp.~2189--2193, March
  2016.

\bibitem{Rappor}
Úlfar Erlingsson, V.~Pihur, and A.~Korolova, ``Rappor: Randomized aggregatable
  privacy-preserving ordinal response,'' in {\em Proceedings of the 21st ACM
  CCCS}, 2014.

\bibitem{Apple}
J.~Tang, A.~Korolova, X.~Bai, X.~Wang, and X.~Wang, ``Privacy loss in apple's
  implementation of differential privacy on {MacOS} 10.12,'' {\em CoRR},
  vol.~abs/1709.02753, 2017.

\bibitem{Tianhao}
T.~Wang, J.~Blocki, N.~Li, and S.~Jha, ``Locally differentially private
  protocols for frequency estimation,'' in {\em 26th {USENIX} Security 17},
  pp.~729--745, {USENIX} Association, 2017.

\bibitem{Lowerbound}
T.-H.~H. Chan, E.~Shi, and D.~Song, ``Optimal lower bound for differentially
  private multi-party aggregation,'' in {\em Proceedings of the 20th Annual
  ECA}, ESA'12, pp.~277--288, 2012.

\bibitem{Raef}
R.~Bassily, K.~Nissim, U.~Stemmer, and A.~Thakurta, ``Practical locally private
  heavy hitters,'' {\em CoRR}, vol.~abs/1707.04982, 2017.

\bibitem{context}
C.~Huang, P.~Kairouz, X.~Chen, L.~Sankar, and R.~Rajagopal, ``Context-aware
  generative adversarial privacy,'' {\em CoRR}, vol.~abs/1710.09549, 2017.

\bibitem{DBLP:journals/corr/WangYZ14b}
W.~Wang, L.~Ying, and J.~Zhang, ``On the tradeoff between privacy and
  distortion in differential privacy,'' {\em CoRR}, vol.~abs/1402.3757, 2014.

\bibitem{ITP1}
S.~Asoodeh, F.~Alajaji, and T.~Linder, ``Notes on information-theoretic
  privacy,'' in {\em 2014 52nd Allerton}, pp.~1272--1278, Sept 2014.

\bibitem{ExtendingDP}
G.~Wu and X.~Xia, ``Extending differential privacy for treating dependent
  records via information theory,'' {\em CoRR}, vol.~abs/1703.07474, 2017.

\bibitem{7498650}
W.~Wang, L.~Ying, and J.~Zhang, ``On the relation between identifiability,
  differential privacy, and mutual-information privacy,'' {\em IEEE
  Transactions on Information Theory}, vol.~62, pp.~5018--5029, Sept 2016.

\bibitem{MIP2}
S.~Asoodeh, F.~Alajaji, and T.~Linder, ``On maximal correlation, mutual
  information and data privacy,'' in {\em IEEE CWIT}, pp.~27--31, July 2015.

\bibitem{Geo}
M.~E. Andr{\'{e}}s, N.~E. Bordenabe, and K.~Chatzikokolakis,
  ``Geo-indistinguishability: Differential privacy for location-based
  systems,'' {\em CoRR}, vol.~abs/1212.1984, 2012.

\bibitem{7930028}
Y.~Cao, M.~Yoshikawa, Y.~Xiao, and L.~Xiong, ``Quantifying differential privacy
  under temporal correlations,'' in {\em 2017 IEEE 33rd ICDE}, pp.~821--832,
  April 2017.

\bibitem{bound}
F.~Tram\`{e}r and Z.~Huang, ``Differential privacy with bounded priors:
  Reconciling utility and privacy in genome-wide association studies,'' in {\em
  Proceedings of the 22Nd ACM SIGSAC}, pp.~1286--1297, 2015.

\bibitem{NIPS2014_5392}
P.~Kairouz, Oh, and Sewoong, ``Extremal mechanisms for local differential
  privacy,'' in {\em Advances in Neural Information Processing Systems 27},
  pp.~2879--2887, 2014.

\bibitem{Centrlized_IP}
F.~Calmon and N.~Fawaz, ``Privacy against statistical inference,'' 10 2012.

\bibitem{MMSE}
F.~A. S.~Asoodeh and T.~Linder, ``Privacy-aware mmse estimation,'' in {\em 2016
  IEEE ISIT}, pp.~1989--1993, July 2016.

\bibitem{heavyhitter}
Z.~Qin, Y.~Yang, and T.~Yu, ``Heavy hitter estimation over set-valued data with
  local differential privacy,'' in {\em Proceedings of the 2016 ACM SIGSAC},
  CCS '16, pp.~192--203, 2016.

\bibitem{papoulis2002probability}
A.~Papoulis and S.~Pillai, {\em Probability, random variables, and stochastic
  processes}.
\newblock McGraw-Hill, 2002.

\end{thebibliography}
\end{document}